\documentclass[5p,10pt]{elsarticle}
\usepackage{amsmath,amssymb,amsfonts}
\usepackage{amsthm}
\usepackage{algorithmic}
\usepackage{graphicx}
\usepackage{textcomp}
\usepackage{mathtools}
\usepackage{xurl}

\usepackage{caption}
\biboptions{sort&compress}

\newtheorem{lemma}{Lemma}
\newtheorem*{lemma*}{Lemma}

\DeclareMathOperator{\divergence}{div}

\DeclareMathOperator{\support}{supp}
\DeclareMathOperator{\sign}{sign}
\DeclareMathOperator{\conv}{conv}

\begin{document}
\begin{frontmatter}
\title{The Localized Subtraction Approach For EEG and MEG Forward Modeling}

\author[ibb_label]{Malte B. Höltershinken\corref{corresponding_author_ref}}
\ead{m\_hoel20@uni-muenster.de}

\author[ibb_label,imi_label]{Pia Lange}

\author[ibb_label]{Tim Erdbrügger}

\author[ibb_label,occ_label]{Yvonne Buschermöhle}

\author[upjv_label,chup_label]{Fabrice Wallois}

\author[tum_label]{Alena Buyx}

\author[tuni_label]{Sampsa Pursiainen}

\author[umit_label]{Johannes Vorwerk}

\author[wwu_label]{Christian Engwer}

\author[ibb_label,occ_label]{Carsten H. Wolters}

\cortext[corresponding_author_ref]{Corresponding author at: Malmedyweg 15, 48149 Münster, Germany}

\affiliation[ibb_label]{
organization={Institute for Biomagnetism and Biosignalanalysis, University of Münster},
city={Münster},
country={Germany}
}

\affiliation[imi_label]{
organization={Institute of Medical Informatics, University of Münster},
city={Münster},
country={Germany}
}

\affiliation[occ_label]{
organization={Otto Creutzfeldt Center for Cognitive and Behavioral Neuroscience, University of Münster},
city={Münster},
country={Germany}
}

\affiliation[upjv_label]{
organization={INSERM U1105, Research Group on Multimodal Analysis of Brain Function, Jules Verne University of Picardie},
city={Amiens},
country={France}
}

\affiliation[chup_label]{
organization={Pediatric functional exploration of the nervous system department, CHU Picardie},
city={Amiens},
country={France}
}

\affiliation[tum_label]{
organization={Institute of History and Ethics in Medicine, Technical University of Munich},
city={Munich},
country={Germany}
}

\affiliation[tuni_label]{
organization={Computing Sciences Unit, Faculty of Information Technology and Communication Sciences, Tampere University},
city={Tampere},
country={Finland}
}

\affiliation[umit_label]{
organization={Institute of Electrical and Biomedical Engineering, Private University for Health Sciences, Medical Informatics and Technology},
city={Hall in Tyrol},
country={Austria}
}

\affiliation[wwu_label]{
organization={Faculty of Mathematics and Computer Science, University of Münster},
city={Münster},
country={Germany}
}

\begin{abstract}
In FEM-based EEG and MEG source analysis, the subtraction approach has been proposed to simulate sensor measurements generated by neural activity. While this approach possesses a rigorous foundation and produces accurate results, its major downside is that it is computationally prohibitively expensive in practical applications. To overcome this, we developed a new approach, called the localized subtraction approach.
This approach is designed to preserve the mathematical foundation of the subtraction approach, while also leading to sparse right-hand sides in the FEM formulation, making it efficiently computable. We achieve this by introducing a cut-off into the subtraction, restricting its influence to the immediate neighborhood of the source.
In this work, this approach will be presented, analyzed, and compared to other state-of-the-art FEM right-hand side approaches. Furthermore, we discuss how to arrive at an efficient and stable implementation.
We perform validation in multi-layer sphere models where analytical solutions exist. There, we demonstrate that the localized subtraction approach is vastly more efficient than the subtraction approach. Moreover, we find that for the EEG forward problem, the localized subtraction approach is less dependent on the global structure of the FEM mesh, while the subtraction approach is sensitive to the mesh structure even at a distance to the source.
Additionally, we show the localized subtraction approach to rival, and in many cases even surpass, the other investigated approaches in terms of accuracy. Especially for the MEG forward problem, where Biot-Savart's law is employed, we show the localized subtraction approach and the subtraction approach to produce highly accurate approximations of the volume currents close to the source.
The localized subtraction approach thus reduces the computational cost of the subtraction approach to an extent that makes it usable in practical applications without sacrificing the rigorousness and accuracy the subtraction approach is known for.
\end{abstract}

\end{frontmatter}

\section{Introduction}
\label{sec:introduction}
In electro- (EEG) and magnetoencephalography (MEG), the so-called EEG and MEG forward problems consist of simulating the sensor measurements that a given neural activity would generate \cite{Haem_MEG, handbook_of_neural_activity_measurement}.
The neural activity is typically modeled as a linear combination of mathematical point dipoles, see e.g. \cite{Haem_MEG, handbook_of_neural_activity_measurement, demunck_dipoles_adequate}. The problem of forward modeling is thus the computation of electric potentials (EEG) and magnetic fluxes (MEG) evoked by a single point dipole. Due to the complex geometry of the head, an accurate approximation of potentials and fluxes relies on a realistic volume conductor model, where one then employs some numerical method \cite{Haem_MEG, handbook_of_neural_activity_measurement, ramon_realistic_head_models}. There are various approaches to the computation of these approximations, e.g. boundary element methods (BEM) \cite{fuchs_bem, makarov_bem_fmm, clerc_sym_bem, acar_bem} or finite difference methods (FDM) \cite{montes_restrepo_fdm, turovets_fdm, cuartas_fdm}. The finite element method (FEM) has been proposed for modeling complex head geometries and tissue anisotropy \cite{handbook_of_neural_activity_measurement, lohrengel_fem, vermaas_femfuns} and is therefore in the focus of this work. A central difficulty for FEM, however, is the embedding of the highly singular point dipole \cite{beltrachini_analytic_subtraction, drechsler_subtraction_fem}.  We call a strategy for this embedding a \textit{FEM potential approach}. 
Potential approaches can be roughly divided into so-called \textit{direct approaches} and \textit{subtraction approaches}. In direct approaches, the point dipole is directly incorporated into the FEM right-hand side. State-of-the-art direct approaches, such as the multipolar Venant approach \cite{vorwerk_venant} and the H(div) approach \cite{miinalainen_eeg_source_modeling}, achieve this by substituting the point dipole with a regularized dipole-like object. Subtraction approaches on the other hand analytically handle the dipole by ``subtracting'' the singularity out of the problem formulation, arriving at a regular problem, and then post-process the resulting FEM solution to add the singularity back in. A substantial amount of work has gone into the development and evaluation of subtraction approaches \cite{thevenet_subtraction, van_den_broek_subtraction, awada_subtraction, marin_subtraction, schimpf_dipole_models, wolters_subtraction_method, drechsler_subtraction_fem, engwer_dg_fem_eeg, beltrachini_projected_subtraction, beltrachini_analytic_subtraction}. In these studies, subtraction approaches were shown to  posess a rigorous mathematical foundation and to produce highly accurate results. Despite this, they are typically not used in practical applications.

The main reason for this is their high computational cost. In source analysis, one typically needs to compute forward solutions for tens of thousands of source positions, and, even when employing a fast transfer matrix approach \cite{gencer_transfer_matrices},  \cite{drechsler_subtraction_fem}, the time it takes to assemble all FEM right-hand sides for subtraction-type approaches in realistic head models becomes computationally infeasible. Different approaches have been suggested to remedy this drawback of subtraction methods. The projected subtraction approach \cite{wolters_subtraction_method} and the projected gradient subtraction approach \cite{beltrachini_projected_subtraction} substitute complicated integrals with simpler ones and thus try to strike a balance between accuracy and computational burden. The analytic subtraction approach \cite{beltrachini_analytic_subtraction} on the other hand derives analytical formulas for the FEM right-hand side and thus removes the need for an expensive numerical calculation. While these approaches all reduce the computational cost of the full subtraction approach as detailed in \cite{drechsler_subtraction_fem}, none of them has succeeded in making subtraction methods viable in practice. 

The fundamental reason for this is that all previous subtraction approaches lead to FEM formulations that produce dense right-hand sides. Assembling these right-hand sides requires an iteration over the whole mesh for every dipole under consideration, where for each element an integration needs to be performed. Especially for realistic head models, often consisting of millions of elements (e.g. \cite{piastra_dataset}), this leads to the aforementioned high computational costs. State-of-the-art direct approaches on the other hand produce sparse right-hand sides, befitting the focality of a point dipole. These sparse right-hand sides can then be assembled rapidly. 

Aside from the high computational complexity, another problem of the different subtraction approaches arises for the MEG forward problem. When utilizing a finite element method to compute magnetic fluxes, one typically first computes an approximation to the electric potential and then employs the Biot-Savart law (see \cite{Haem_MEG}). Using a direct approach, this poses no difficulty, as the numerical solution directly yields the electric current in a form suitable for the evaluation of Biot-Savart's law. In contrast, subtraction approaches produce potentials that are singular at the dipole position, and since Biot-Savart's law requires integration over the whole head volume, a direct application of this law leads to integrals over singular functions, which is numerically undesirable.

In this paper, we build on top of these observations and propose a new potential approach, called the \textit{localized subtraction} approach, that will alleviate the aforementioned problems. The construction of this approach is based on the point dipole being only singular in the source position. Hence we can subtract the singularity out of the problem formulation in a spatially confined manner, in contrast to previous subtraction approaches, where this subtraction was performed on the whole head domain.  This localization of the subtraction is the central step in reducing the computational burden of subtraction approaches.

Furthermore, we present an approach for handling the singular integral arising in the MEG forward problem when employing a subtraction-type potential approach. Building on top of the derivation of the classical Geselowitz formulas \cite{geselowitz_meg}, the singular part is extracted out of the integral and rewritten in terms of non-singular surface integrals.

In the methods section, we will derive the corresponding formulas for the EEG and MEG cases. In particular, we see that for the EEG forward problem the localized subtraction approach produces right-hand sides that are only nonzero for degrees of freedom in the close proximity of the dipole, leading to sparse right-hand sides. This renders the FEM right-hand sides efficiently computable. In the results section, we will then show that the localized subtraction approach yields a considerable increase in speed when compared to previous subtraction approaches. Furthermore, we will show the localized subtraction approach to be similarly accurate, and in many cases even more accurate, than other state-of-the-art potential approaches, such as the analytical subtraction approach \cite{beltrachini_analytic_subtraction} and the multipolar Venant approach \cite{vorwerk_venant}.

\section{Methods}
\subsection{The EEG forward problem}
\subsubsection{Definition of the problem}
Let $\Omega$ be the head domain and $\partial \Omega$ its boundary. In the EEG forward problem we are interested in determining the electric potential $u : \Omega \rightarrow \mathbb{R}$, or at least its values at some predetermined electrode positions, due to some given neural activity. As described in the introduction, in current practice this reduces to the computation of potentials generated by mathematical point dipoles. Let $\sigma : \Omega \rightarrow \mathbb{R}^{3\times 3}$ be the symmetric positive definite conductivity tensor. We now consider a point dipole at position $x_0 \in \Omega$ with moment $M \in \mathbb{R}^3$. Using a quasistatic approximation of Maxwell's equations \cite{Haem_MEG}, the electric potential $u$ due to the point dipole can be described by
\begin{align}
\divergence\left(\sigma\nabla u\right) &= \divergence\left( M \cdot \delta_{x_0}\right) &\text{on $\Omega$,} \label{pde_volume}\\
\langle \sigma \nabla u, \eta\rangle &= 0 &\text{on $\partial\Omega$,} \label{pde_boundary}
\end{align}
where $\delta_{x_0}$ is the Dirac distribution at $x_0$ and $\eta$ is the unit outer normal of the head domain.

\subsubsection{Deriving a weak formulation}
When trying to solve a partial differential equation using a finite element approach, the first step consists of deriving a suitable weak formulation. In doing so, we have to deal with the highly singular term $\divergence(M \cdot \delta_{x_0})$ arising from (\ref{pde_volume}). The defining feature of subtraction approaches is that they do not directly incorporate this term, or an approximation of it, into the weak formulation, but instead use (\ref{pde_volume}) and (\ref{pde_boundary}) to derive a weak formulation for the function $u - \varphi$, where $\varphi$ is some suitable function which ``subtracts'' the singularity out of the problem. We will now discuss how such a $\varphi$ can be constructed.

We assume the conductivity $\sigma$ to be constant on a neighborhood of $x_0$. Let $\sigma^\infty \in \mathbb{R}^{3\times3}$ be its value on this neighborhood. Now consider the equation
\begin{align}
\divergence\left(\sigma^\infty \nabla u^\infty\right) &= \divergence\left( M \cdot \delta_{x_0}\right) &\text{on $\mathbb{R}^3$,} \label{pde_whole_space}
\end{align}
which is the equivalent of (\ref{pde_volume}) in an unbounded homogeneous conductor. In this context, it is straightforward to derive analytical formulas for $u^\infty$ (e.g. \cite{drechsler_subtraction_fem, wolters_subtraction_method}). Since the right-hand side in (\ref{pde_whole_space}) is the same as in (\ref{pde_volume}), the idea is to use this $u^\infty$ to eliminate the singularity. 

One might now simply set $\varphi = u^\infty$. All previous subtraction methods used this approach, and this is what led to their high computational demand. The reason for this is that $u^\infty$ is non-zero almost everywhere. Hence, weak formulations, and in a second step FEM discretizations, derived from this approach contain integrals against functions that are non-zero on a large portion of the head domain \cite{drechsler_subtraction_fem, wolters_subtraction_method}. When evaluating these integrals for the FEM basis functions, a large number of them will be non-zero, leading to dense right-hand sides.

We instead propose to construct $\varphi$ in a different manner. Since $\divergence\left(M \cdot \delta_{x_0}\right)$ is only singular at the source position $x_0$, we suggest to only use the local distribution of $u^\infty$ around the source position to remove the singularity from the problem. We make this more precise. Let $\chi : \Omega \rightarrow \mathbb{R}$ be a function with $\chi = 1$ on a neighborhood of the source position $x_0$. The conceptual idea is to choose $\chi$ in such a way that it is only non-zero in a small region around $x_0$. 
Then $\chi \cdot u^\infty$ is also only non-zero in a small region around $x_0$, ultimately leading to a sparse right-hand side. Furthermore, it shares the local behavior of $u^\infty$ close to the source position and can hence also be used to eliminate the singularity from the problem formulation. 
Depending on whether one aims for a (classical) continuous Galerkin or a discontinuous Galerkin \cite{piastra_dgfem_eeg_meg} finite element discretization, one has to impose different regularity constraints on $\chi$. We focus on the continuous Galerkin case and refer to \cite{phd_nuessing} for the discontinuous Galerkin case. In the continuous Galerkin case, we further require $\chi$ to be continuous and to be contained in the Sobolev space $H^1(\Omega)$. We now make the approach $\varphi = \chi \cdot u^\infty$, and define the \textit{correction potential} $u^c$ by
\begin{equation}
u^c \coloneqq u - \chi \cdot u^\infty.
\end{equation}

We now derive a weak formulation for this correction potential. The central step is to understand how the multiplication of $u^\infty$ with $\chi$ impacts the distributional derivative. This is answered by the following lemma, which is proven in \ref{proof_of_distributional_derivative}.
\begin{lemma}
\label{distributional_derivative_u_infty_times_chi}
Let $\Omega^\infty \subset \Omega$ and $\chi \in H^1(\Omega)$ be such that $\chi$ is identical to $1$ on $\Omega^\infty$ and $x_0 \in \Omega^\infty$. Let $\widetilde{\Omega} \subset \Omega\setminus\Omega^\infty$ be a region such that $\partial\Omega^\infty$ is a subset of $\partial\Omega \cup \partial\widetilde{\Omega}$, and $\chi = 0$ outside $\Omega^\infty \cup \widetilde{\Omega}$ (see Figure \ref{patch_visualization}). Then for any test function $\phi \in C^\infty_c(\Omega)
\footnote{By $C^\infty_c(\Omega)$ we denote the space of all infinitely differentiable functions with compact support in $\Omega$.}$
 we have in the distributional sense
\begin{multline}
\divergence\left(\sigma^\infty \nabla \left(\chi \cdot u^\infty \right)\right)(\phi) 
 = \\ \divergence\left(M \cdot \delta_{x_0}\right)(\phi)
 - \int_{\widetilde{\Omega}} \langle\sigma^\infty \nabla \left(\chi \cdot u^\infty \right), \nabla \phi \rangle\, dV \\
 - \int_{\partial\Omega^\infty} \langle\sigma^\infty \nabla u^\infty , \eta\rangle \phi \, dS.
\end{multline}
\end{lemma}
Here $\eta$ is the unit outer normal of $\Omega^\infty$. Note in particular that we have $\divergence\left(M \cdot \delta_{x_0}\right)$ on the right-hand side, and hence can use $\chi \cdot u^\infty$ to cancel the singularity in \eqref{pde_volume}.

\begin{figure}[!t]
\centerline{\includegraphics{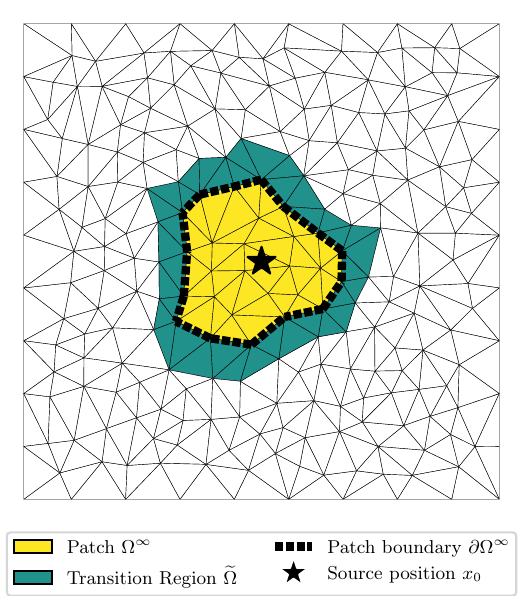}}
\caption{2-dim. visualization of a realization of the construction described in lemma \ref{distributional_derivative_u_infty_times_chi}. The whole square is $\Omega$. In the yellow region we have $\chi = 1$, in the white region we have $\chi = 0$.}
\label{patch_visualization}
\end{figure}

We call $\Omega^\infty$ the \textit{patch}. Furthermore, we call $\widetilde{\Omega}$ the \textit{transition region}, as $\chi$ transitions from 1 to 0 here. We want to emphasize that our definition of $\Omega^\infty$ differs from the one given in \cite{drechsler_subtraction_fem, wolters_subtraction_method}. In particular, we do not assume the conductivity $\sigma$ to be constant on the whole patch $\Omega^\infty$. Instead, we just assume that $\Omega^\infty$ contains a neighborhood of $x_0$ where the conductivity is constant. 

Now let $\sigma^c \coloneqq \sigma - \sigma^\infty$. Multiplying \eqref{pde_volume} by a test function, applying partial integration together with the boundary condition \eqref{pde_boundary}, inserting the decomposition 
\begin{equation}
\sigma \nabla u = \sigma \nabla u^c + \sigma^c \nabla \left (\chi \cdot u^\infty \right) + \sigma^\infty \nabla \left(\chi \cdot u^\infty \right)
\end{equation}
and using lemma \ref{distributional_derivative_u_infty_times_chi}, we arrive at the following.
Define a bilinear form $a : H^1(\Omega) \times H^1(\Omega) \rightarrow \mathbb{R}$ by
\begin{equation}
a(w, v) = \int_{\Omega} \langle \sigma \nabla w, \nabla v \rangle\, dV
\end{equation}
and a linear form $l : H^1(\Omega) \rightarrow \mathbb{R}$ by
\begin{multline}
l(v) = - \int_{\widetilde{\Omega}} \langle \sigma \nabla \left(\chi \cdot u^\infty \right), \nabla v\rangle \, dV \\
- \int_{\partial \Omega^\infty} \langle\sigma^\infty \nabla u^\infty, \eta \rangle v \, dS
- \int_{\Omega^\infty} \langle \sigma^c \nabla u^\infty, \nabla v \rangle \, dV,
\label{localized_subtraction_rhs_main_paper}
\end{multline}
where $\eta$ is the unit outer normal of $\Omega^\infty$.
Then the \textit{continuous Galerkin localized subtraction approach} is given by the problem of finding $u^c \in H^1(\Omega)$ such that for all $v \in H^1(\Omega)$ we have $a(u^c, v) = l(v)$. For a more detailed derivation, we refer to \ref{weak_formulation_correct}.

It is easy to see that $l$ vanishes on constant functions. By standard finite element theory, described e.g. in \cite{wolters_subtraction_method}, one can thus see that the set of solutions of this problem is of the form
\begin{equation}
u^c_0 + \mathbb{R} \cdot 1,
\end{equation} reflecting the arbitrary choice of a reference electrode. We have thus derived a weak formulation and have seen that it produces solutions in the expected form. Also note that the corresponding family of solutions of \eqref{pde_volume} and \eqref{pde_boundary} given by
\begin{equation}
u^c_0 + \chi \cdot u^\infty + \mathbb{R} \cdot 1
\end{equation}
does not depend on the choice of $\chi, \Omega^\infty$ and $\widetilde{\Omega}$, as is to be expected. Proofs of these statements can be found in \ref{well_defined_proof}.
Furthermore, we would like to emphasize that the choice $\chi = 1$, $\Omega^\infty = \Omega$ and $\widetilde{\Omega} = \emptyset$ leads to the classical subtraction approach as described in \cite{thevenet_subtraction, van_den_broek_subtraction, awada_subtraction, marin_subtraction, schimpf_dipole_models, wolters_subtraction_method, drechsler_subtraction_fem, engwer_dg_fem_eeg, beltrachini_projected_subtraction, beltrachini_analytic_subtraction}. Hence the classical subtraction approach can be seen as a special case of the localized subtraction approach.

\subsubsection{FEM discretization}
\label{fem_discretization}
To use \eqref{localized_subtraction_rhs_main_paper} in a computational algorithm, we need a way to construct $\chi, \Omega^\infty$ and $\widetilde{\Omega}$. We want to propose a strategy for this. Since we want to use a FEM approach, we can assume that a mesh of the head domain is given. We then construct $\Omega^\infty$ and $\widetilde{\Omega}$ using so called \textit{vertex extensions}. If we are given a certain subset $\mathcal{T}$ of mesh elements, we define the vertex extension $\mathcal{T}'$ of $T$ to be the set of all mesh elements that share at least one vertex with an element from $\mathcal{T}$. Note that $\mathcal{T} \subseteq \mathcal{T}'$. Now let $x_0 \in \Omega$ be a given source position. We then take the mesh element containing $x_0$ and perform a certain number of vertex extensions. We take the result of this as the patch $\Omega^\infty$. We then perform an additional vertex extension on the patch and define $\widetilde{\Omega}$ to be the union over all elements contained in the vertex extension of the patch that are not contained in the patch itself. We then choose $\chi$ as a linear combination of FEM basis functions which takes the value $1$ on every patch vertex and $0$ on the remaining vertices. Note that in the common cases of piece-wise linear basis functions on tetrahedra and piece-wise multilinear basis functions on hexahedra this already uniquely determines $\chi$. The result of this construction is illustrated in Figure \ref{patch_visualization}. Note that in this figure two vertex extensions were used to construct the patch $\Omega^\infty$.

This construction leaves the question of how many vertex extensions should be used to construct the patch. From a computational perspective, it is beneficial to perform as few vertex extensions as possible since larger patches lead to a more costly integration during the FEM right-hand side assembly. Furthermore, for larger patches the accuracy of the approximation becomes more dependent on the underlying mesh, since e.g. at conductivity jumps $u^c = u - \chi \cdot u^\infty$ might exhibit an intricate structure, requiring a high resolution to properly resolve. On the other hand, choosing small patches might also lead to suboptimal approximations, since in this case $\chi \cdot u^\infty$ is only non-zero in a small vicinity of the source, which forces the singular behavior of the potential near the dipole onto the correction potential $u^c$. Hence it becomes more difficult to approximate $u^c$ by finite element trial functions, thus degrading the performance of the approach. In the Results section, we will investigate how we can optimally balance these effects.

Finally, we want to elaborate on the computation of the FEM right-hand side. Let $K$ be a mesh element and $F$ a face of a mesh element. When the right-hand side assembly is performed element-wise, as is typical in the FEM, and $\varphi$ is a FEM test function, we need to compute integrals of the form
\begin{align}
I_T =& \int_K \langle \sigma \nabla \left(\chi \cdot u^\infty \right), \nabla \varphi\rangle \, dV \label{transition_integral_eeg_main_paper}\\
I_S =& \int_F \langle \sigma^\infty \nabla u^\infty , \eta \rangle \varphi \, dS, \label{surface_integral_eeg_main_paper}\\
I_P =& \int_K \langle\sigma^c \nabla u^\infty, \nabla \varphi \rangle \, dV, \label{patch_integral_eeg_main_paper}
\end{align} 
We call $I_T$ a \textit{transition integral}, $I_S$ a \textit{surface integral} and $I_P$ a \textit{patch integral}.

In the case of isotropic $\sigma^\infty$ and piecewise affine trial functions on tetrahedral meshes, analytical expressions for the patch and surface integrals were derived in \cite{beltrachini_analytic_subtraction}, yielding the analytical subtraction approach. We slightly adapted the corresponding formulas to arrive at a more stable implementation. Furthermore, building on top of \cite{wilton_analytic_formulas, graglia_analytic_formulas}, we extended the derivation given in \cite{beltrachini_analytic_subtraction} to compute an analytical expression for the transition integral, enabling a completely analytic computation of the localized subtraction EEG right-hand side. We refer to \ref{analytical_expressions} for an in-depth derivation of these formulas.

If the mesh is not tetrahedral, or the conductivity in the neighborhood of the source is not isotropic, or the trial functions are not affine, we resort to numerical integration using the quadrature rules implemented in the DUNE framework \cite{bastian_dune_framework}, on which the DUNEuro implementation is built \cite{Schrader_duneuro}. A discussion on how to choose suitable integration orders can be found in \ref{integration_order_chapter}.

\subsection{The MEG forward problem}
\subsubsection{Definition of the problem}
In magnetoencephalography, one measures magnetic fluxes through magnetometers. To simulate magnetic fluxes, values of the form
\begin{equation}
\int_F \langle B, \eta\rangle\, dS, \label{meg_sensor_measurement}
\end{equation}
need to be approximated, where $B$ is the magnetic field, $F$ is some surface enclosed by a pickup coil and $\eta$ is the normal to this surface (see e.g. \cite{Haem_MEG, handbook_of_neural_activity_measurement}). In FEM approaches, this is typically done using some quadrature rule on the integral \eqref{meg_sensor_measurement}, and then using Biot-Savart's law to compute the magnetic field $B$ at the quadrature points \cite{roth_meg_integration_formulas}. This law states that, for the electric current $j$ in the volume conductor $\Omega$, the magnetic field at some point $x$ outside $\Omega$ is given by
\begin{equation}
B(x) = \frac{\mu_0}{4\pi} \int_\Omega j(y) \times \frac{x - y}{\|x - y\|^3} \, dV(y), \label{biot_savart_law}
\end{equation}
where $\mu_0$ is the magnetic permeability of the vacuum \cite{Haem_MEG, handbook_of_neural_activity_measurement}.
In bioelectromagnetism, the current is typically split into the so-called \textit{primary current $j^P$} generated by the neural activity, and the passive volume currents generated by the electric field and hence given by $- \sigma\nabla u$, where $u$ is the electric potential \cite{Haem_MEG, handbook_of_neural_activity_measurement}. Since the magnetic field depends linearly on the primary current, we can again concentrate on the case of a point dipole $j^P = M \cdot \delta_{x_0}$. Inserting the splitting $j = j^P - \sigma \nabla u$ into \eqref{biot_savart_law}, we get a splitting $B = B^P + B^S$ into the so-called \textit{primary magnetic field $B^P$} and \textit{secondary magnetic field $B^S$}. Here, the primary field can be efficiently computed using analytical expressions and hence poses no problem. The difficult part of the MEG forward problem is thus to compute the integral
\begin{equation}
\int_{\Omega} \sigma(y) \nabla u(y) \times  \frac{x - y}{\|x - y\|^3} \, dV(y). \label{secondary_magnetic_field}
\end{equation}

\subsubsection{The boundary subtraction approach}
To compute the integral in \eqref{secondary_magnetic_field}, we need the gradient of the electric potential $u$, which we can e.g. approximate using a FEM approach. With an increasing distance to the source, the magnitude of the volume current decreases. Hence the behavior of $\sigma \nabla u$ close to the source position has a comparatively high contribution in the integral \eqref{secondary_magnetic_field}, and it is of interest to ensure that the FEM approximation is accurate in the vicinity of the dipole. The local behavior of the numerical solution depends on how the FEM right-hand side was derived from the dipole, which means it depends on the potential approach. Since the current state-of-the-art direct approaches, such as the Venant approach, are based on substituting the dipole with a more easily manageable object, their local behavior around the source position by construction deviates from that of a true dipole. Subtraction methods in contrast incorporate the mathematical point dipole into the problem formulation and can thus be expected to produce highly accurate approximations of the true current density in the proximity of the dipole. This is illustrated in Figure \ref{flux_loc_sub}, where we see that the current computed using the localized subtraction approach has the expected dipolar pattern even in the close vicinity of the source position. Hence, when employing a dipolar source model for the neural activity, subtraction approaches are a natural choice for computing currents and their resulting magnetic fields.

\begin{figure}[!t]
\centerline{\includegraphics[scale = 0.2, trim = 2cm 2.5cm 1.5cm 2.0cm, clip]{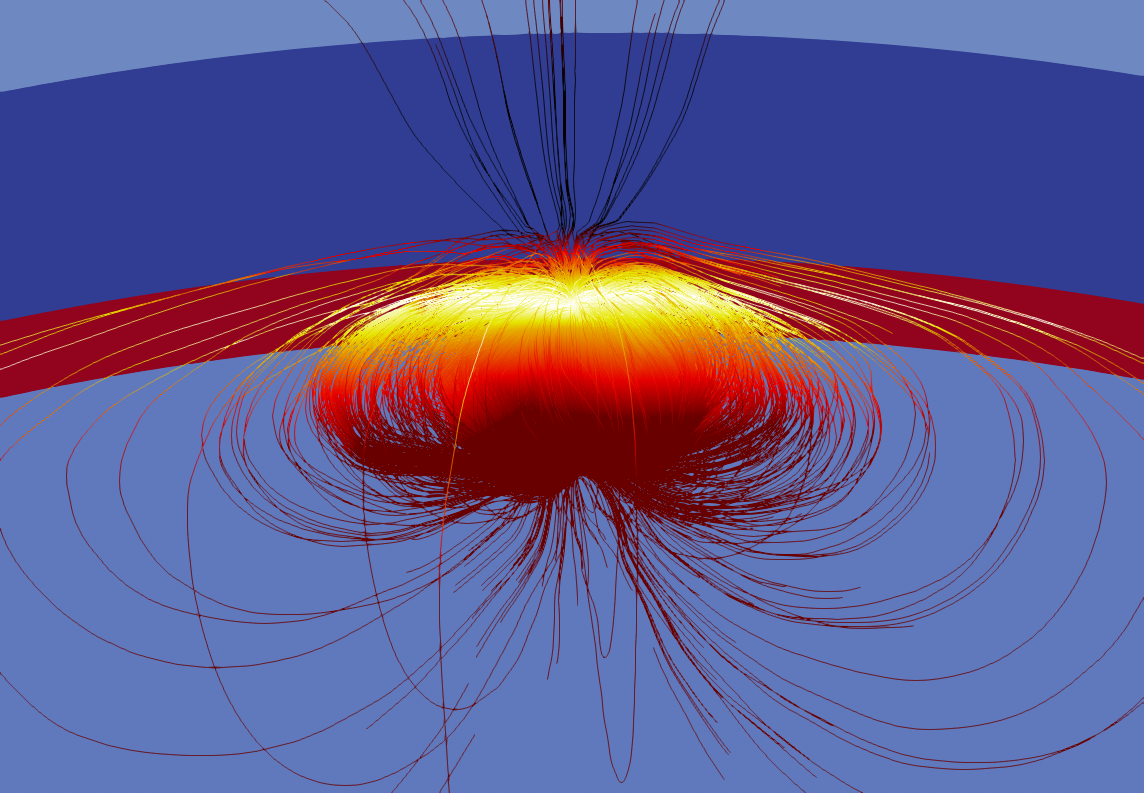}}
\caption{Visualization of the volume currents as computed by the localized subtraction approach. The colors of the lines indicate the magnitude of the current and the colors in the background show the conductivity of the volume conductor. The numerical simulation was performed in the multilayer sphere model \textit{mesh\_init} shown in Figure \ref{meshes_refinement}\,(a), for a radial dipole with a 1 mm distance to the conductivity jump. Details on this mesh can be found in the Results section. The current was visualized by generating 5000 points in a sphere of radius 2.5 mm around the dipole position and applying a Runge-Kutta method, using the streamtrace filter implemented in the ParaView software \cite{ahrens_paraview}.} 
\label{flux_loc_sub}
\end{figure}

This again introduces a singularity into the problem, as subtraction approaches produce approximations of the form $u_h = u_h^c + \chi \cdot u^\infty$, where $u_h^c$ is the result of the finite element discretization. Hence we need to compute
\begin{equation}
\int_\Omega \sigma(y) \nabla \left(u^c_h + \chi u^\infty\right)(y) \times \frac{x - y}{\|x - y\|^3}\, dV(y). \label{singular_formulation_biot_savart}
\end{equation}
Since $u^c_h$ is a linear combination of FEM basis funcitons, it poses no problem for the integration, which can be performed efficiently using a transfer matrix approach \cite{drechsler_subtraction_fem, gencer_transfer_matrices, wolters_subtraction_method}. The function $\nabla \left(\chi \cdot u^\infty\right)$ on the other hand contains a singularity at the source position, so that numerical integration is undesirable. One way to deal with this is projecting either $u^\infty$ or its gradient into a space of easily integrable functions, similar to what is done in \cite{wolters_subtraction_method} or \cite{beltrachini_projected_subtraction}. Since especially in the vicinity of the singularity the function $u^\infty$ is quite different from its projection, this can introduce a significant error, as was observed in \cite{beltrachini_projected_subtraction} in the EEG context. Under the assumption of isotropic $\sigma^\infty$ we can instead eliminate the singularity from \eqref{singular_formulation_biot_savart} altogether. For a set $A$, let $1_A$ denote its indicator function, which is 1 on $A$ and $0$ everywhere else. Then, if $\Omega^\infty$ and $\widetilde{\Omega}$ are the patch and transition region, by construction, we have the decomposition
\begin{multline}
\sigma \nabla\left(\chi u^\infty\right) = \\ 
\sigma \nabla \left(\chi u^\infty\right) \cdot 1_{\widetilde{\Omega}}
+ \sigma^c \nabla u^\infty \cdot 1_{\Omega^\infty}
+ \sigma^\infty \nabla u^\infty \cdot 1_{\Omega^\infty}.
\end{multline}
Note that $\sigma \nabla \left(\chi u^\infty\right) \cdot 1_{\widetilde{\Omega}}$ and $\sigma^c \nabla u^\infty \cdot 1_{\Omega^\infty}$ do not contain a singularity, since $1_{\widetilde{\Omega}}$ and $\sigma^c$ vanish on an environment of the singularity. Thus the singularity of the integrand in \eqref{singular_formulation_biot_savart} can be split off into $\sigma^\infty \nabla u^\infty \cdot 1_{\Omega^\infty}$. But this expression can be handled by mimicking the derivation of the classical Geselowitz formula from \cite{geselowitz_meg}. This idea was introduced by \cite{phd_piastra_meg}, where it is used for MEG forward simulations in the context of the full subtraction approach \cite{drechsler_subtraction_fem}. More concretely, we have the following.
\begin{lemma}
Let $\eta$ be the unit outer normal of $\Omega^\infty$. We then have
\begin{multline}
\int_{\Omega^\infty} \sigma^\infty \nabla u^\infty(y) \times \frac{x - y}{\|x - y\|^3} \, dV(y)
=\\ \int_{\partial \Omega^\infty} \sigma^\infty u^\infty(y) \cdot \eta(y) \times \frac{x - y}{\|x - y\|^3} \, dS(y).
\end{multline}
\end{lemma}
A proof of this lemma can be found in \cite{dassios_geselowitz_formula}, where the authors take special care to properly handle the singularity at the source position. 

In total, we thus arrive at the following approach for computing the secondary magnetic field.
\begin{align}
\begin{split}
-\frac{4\pi}{\mu_0}B^S(x) =&
\int_{\Omega} \sigma(y) \nabla u^c(y) \times \frac{x - y}{\|x - y\|^3} \, dV(y)\\
+& \int_{\Omega^\infty} \sigma^c(y) \nabla u^\infty(y) \times \frac{x - y}{\|x - y\|^3}\, dV(y)\\
+& \int_{\partial \Omega^\infty} \sigma^\infty u^\infty(y) \cdot \eta(y) \times \frac{x - y}{\|x - y\|^3} \, dS(y)\\
+& \int_{\widetilde{\Omega}} \sigma(y) \nabla\left(\chi \cdot u^\infty\right)(y) \times \frac{x - y}{\|x - y\|^3} \, dV(y).
\end{split}
\label{boundary_subtraction}
\end{align}
Note in particular that the right-hand side of \eqref{boundary_subtraction} contains no singularity. We call this the \textit{boundary subtraction approach} for the MEG forward problem since the key step consists in rewriting the singular volume integral as a boundary integral.

The first summand in \eqref{boundary_subtraction} can be assembled using a standard MEG transfer matrix approach. Similar to the EEG case, we call the second summand \textit{patch flux}, the third \textit{surface flux}, and the fourth \textit{transition flux}. These can be assembled in an element-wise manner using numerical integration, which raises the question of how to choose suitable integration orders. A corresponding discussion can be found in \ref{integration_order_chapter}.

\subsection{Implementation}
All formulas derived in the Methods section were implemented into the DUNEuro toolbox \cite{Schrader_duneuro}\footnote{\url{https://gitlab.dune-project.org/duneuro/duneuro/-/tree/feature/localized-assembler}}, an open-source C++ toolbox for neuroscience applications built upon the DUNE framework \cite{bastian_early_dune_paper, bastian_dune_framework}.

\subsection{Experimental setup}
For the numerical simulations in this paper, we used a four-layer sphere model, where the layers represent the brain, cerebrospinal fluid (CSF), skull, and skin. The model parameters are described in table \ref{sphere_parameters}. In multilayer sphere models one can derive analytical solutions for the EEG \cite{demunck_analytical_eeg} and MEG \cite{sarvas_analytical_meg} forward problems, which we compare against our numerical solutions.

\begin{table}
\caption{Multi-layer sphere parameters}
\label{sphere_parameters}
\begin{tabular}{|p{71.5pt}|p{71.5pt}|p{71.5pt}|}
\hline
Compartment & Conductivity & Radius\\ \hline
Brain 		& 0.33 S/m 		& \parbox{\widthof{00}}{0} - 78 mm\\
CSF 			& 1.79 S/m 		&78 - 80 mm\\
Skull 		& 0.01 S/m 		& 80 - 86 mm\\
Skin 			& 0.43 S/m 		& 86 - 92 mm\\ \hline
\end{tabular}
\end{table}

To employ a finite element approach, one needs to construct a mesh. It is well known that the mesh construction itself can already introduce a bias in favor of certain potential approaches \cite{lew_potential_approach_comparison, de_munck_wolters_clerc_handbook_neural_activity}. More concretely, in \cite{lew_potential_approach_comparison} it was found that meshes that concentrate their nodes in the proximity of conductivity jumps favor subtraction approaches, while meshes with evenly distributed nodes favor direct approaches, such as the multipolar Venant approach. 

This can be explained by the fact that the subtraction approach does not solve for the potential itself, but for the difference in the head potential and the potential in an infinite homogeneous medium. To produce accurate results, the FEM needs a high resolution in those areas where the function to be approximated shows a complicated behavior. In particular, in the proximity of conductivity jumps we expect the head potential to deviate from the potential in a homogeneous medium. Hence we expect their difference to exhibit a quite nontrivial behavior at conductivity jumps, which needs a high mesh resolution to properly capture. In the brain layer, on the other hand, we expect the homogeneous potential and the head potential to be quite similar, and hence, as long as the dipole is not too close to a conductivity jump, a comparatively low brain resolution is required to achieve good numerical accuracy for the subtraction approach.

The Venant approach, on the other hand, relies on good node positions for the placement of monopolar loads, and can hence utilize meshes that are more refined inside the brain compartment.

For the localized subtraction approach, similar considerations as for the subtraction approach apply. But in contrast to the subtraction approach, which solves for the whole difference $u - u^\infty$, the localized subtraction approach solves for $u - \chi \cdot u^\infty$, which only differs from $u$ in the proximity of the source. For dipoles with a large distance to the conductivity jump, all potential approaches generally produce highly accurate results, and hence the mesh structure is mainly relevant for the sources close to a conductivity jump. We thus expect the localized subtraction approach to produce accurate results as long as the mesh has a high resolution at the conductivity jumps close to the source positions.

These considerations show the difficulty in choosing a ''fair'' mesh for comparing the accuracies of the different potential approaches. We will use three different meshes, build in such a way as to provide favorable circumstances for the different potential approaches introduced before (see Figure \ref{meshes_refinement}). First, a tetrahedral mesh consisting of 787k nodes was constructed using Gmsh \cite{gmsh_paper} (Figure \ref{meshes_refinement}\,(a)). This mesh has a high concentration of nodes in the CSF and skull compartments. We then used the UG software \cite{bastian_ug_grid} to generate two more meshes from this initial mesh, one by refining the brain compartment (benefitting the Venant approach, Figure \ref{meshes_refinement}\,(b)) and one by refining the skin compartment (benefitting the subtraction approach, Figure \ref{meshes_refinement}\,(c)). In the following, we call the mesh in Figure \ref{meshes_refinement}\,(a) \textit{mesh\_init}, the one in Figure \ref{meshes_refinement}\,(b) \textit{mesh\_brain} and the one in Figure \ref{meshes_refinement}\,(c) \textit{mesh\_skin}.

\begin{figure}
\includegraphics[scale = 0.16]{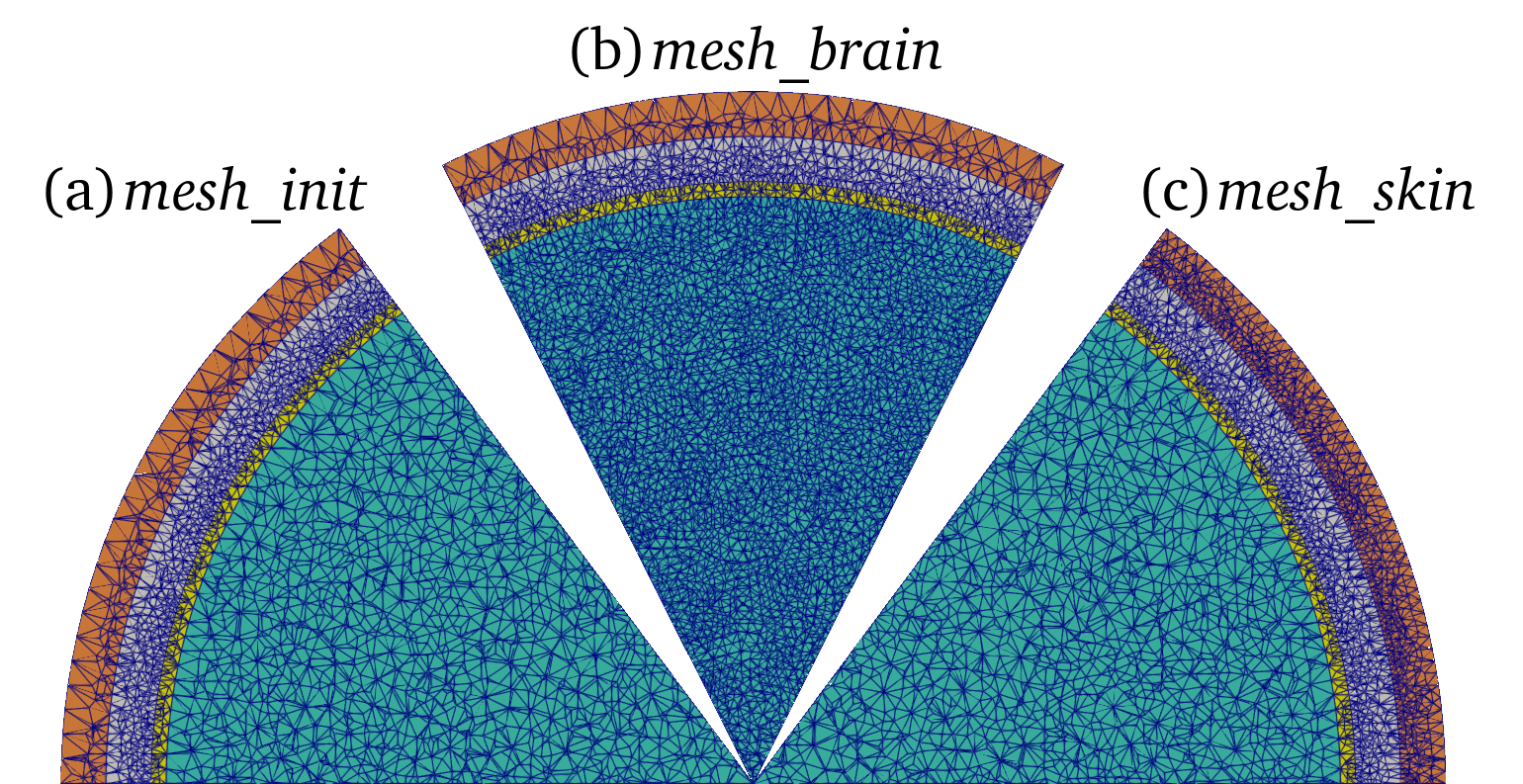}
\caption{Circular sections of the meshes used in the numerical tests. (a) shows the initial mesh with a high concentration of nodes in the CSF and skull compartment. (b) shows the mesh from (a) after refining the brain compartment. (c) shows the mesh from (a) after refining the skin compartment. Mesh (a) consists of ca. 800k nodes, mesh (b) consists of ca. 1.8 million nodes and mesh (c) consists of ca. 1.3 million nodes.}
\label{meshes_refinement}
\end{figure}

For the EEG forward problem we used 200 electrode positions on the sphere surface, and for the MEG forward problem we used 256 coil positions at a distance of 110 mm to the sphere center. The electrode positions and the coil positions respectively were chosen to be approximately uniformly distributed on their corresponding sphere. We always considered dipoles in the brain compartment, which we grouped by \textit{eccentricity}, which we computed as the distance of the dipole from the center of the sphere divided by the radius of the brain compartment. Based on physiological reasoning, \cite{piastra_dgfem_eeg_meg} argues that sources with a distance of 1-2 mm from the CSF compartment are the most relevant for the generation of neural electromagnetic fields. We will thus particularly focus on eccentricities inside this range. Furthermore, it is well known that the numerical simulation of potentials and fluxes becomes more difficult for sources close to a conductivity jump, as is e.g. noted in \cite{wolters_subtraction_method}. It can now happen, e.g. due to segmentation inaccuracies, that sources are placed at a distance of below 1 mm to a conductivity jump. To also cover these extreme cases, we test sources up to an eccentricity of 0.99, corresponding to a distance of about 0.8 mm to the conductivity jump.

\section{Results and Discussion}
\subsection{Patch construction}
\label{patch_construction_subsection}
We first want to investigate the number of vertex extensions one should perform during the construction of the patch for the localized subtraction approach (see section \ref{fem_discretization} and Figure \ref{patch_visualization}). We generated 1000 tangential and 1000 radial dipoles at an eccentricity of $0.99$, each approximately uniformly distributed on the corresponding sphere. For each of the three meshes shown in Figure \ref{meshes_refinement}, we then performed an increasing number of vertex extensions for each dipole and computed the corresponding localized subtraction EEG FEM solution. We also computed the FEM solutions for the analytical subtraction approach from \cite{beltrachini_analytic_subtraction} for comparison. Once the numerical results were computed, the relative error with respect to the analytical solution at the electrode positions was calculated. The results for the tangential dipoles and \textit{mesh\_init} are shown in Figure \ref{vertex_extension_tangential_comparison}.

\begin{figure}[h]
\centerline{\includegraphics[trim = 0.0cm 0cm 0cm 0cm]{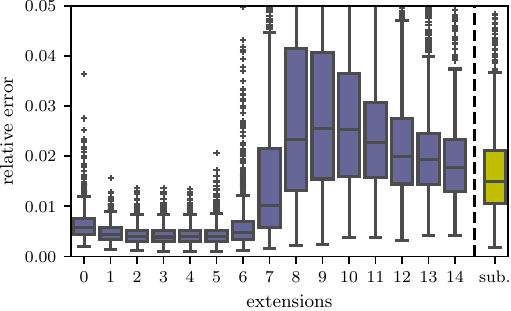}}
\caption{Relative error in the EEG case for 1000 tangential dipoles at 0.99 eccentricity for different numbers of vertex extensions during patch construction, computed using \textit{mesh\_init} (see Figure \ref{meshes_refinement}\,(a)). The rightmost yellow boxplot shows the errors when employing the analytical subtraction approach from \cite{beltrachini_analytic_subtraction}.} 
\label{vertex_extension_tangential_comparison}
\end{figure}

We see that when going from $0$ extensions, meaning the patch consists only of the element containing the dipole, to two extensions, the error decreases. It then stays almost constant up to 5 extensions. The error then significantly increases from 6 to 8 extensions, and then gradually approaches the error of the analytical subtraction approach. Note that if enough vertex extensions are performed to cover the whole volume conductor, the error of the localized subtraction and  the analytical subtraction approach would be identical. Furthermore, radial dipoles show similar behavior.

The initial decrease can be explained by the singular behavior in the correction potential, which comes from the need to approximate the dipole potential in the close proximity of the singularity. Furthermore, for \textit{mesh\_init}, 6-8 extensions are the point where the patches for dipoles at an eccentricity of $0.99$ start to grow into the skin compartment. As seen in Figure \ref{meshes_refinement}\,(a), the skin compartment of this mesh has a comparatively low resolution. This makes an approximation of $u - \chi \cdot u^\infty$,  which might exhibit quite a complicated behavior at conductivity jumps, difficult. Indeed, when performing the same experiment for \textit{mesh\_skin}, we see that the error for larger patches significantly decreases, while the error for smaller patches barely changes (see Figure \ref{vertex_extension_comparison_skin_refined_tangential_eeg}). In contrast, if one performs this experiment for \textit{mesh\_brain} and compares against Figure \ref{vertex_extension_tangential_comparison}, the error for larger patches barely changes, while the error for smaller patches decreases (see Figure \ref{vertex_extension_comparison_brain_refined_tangential_eeg}). The reason for this is that small patches can exploit the higher resolution in the brain area, while larger patches are still limited by the low skin resolution.

Additionally, when comparing the FEM solutions from a small patch with those from a large patch, we see that the relative difference is high in the skin elements close to the source and adjacent to the skull. This also explains why, after the increase in the error at 6-9 extensions, the error continuously drops, since with growing patches we need to approximate $u - u^\infty$, and not the more complicated $u - \chi \cdot u^\infty$, in the skin elements close to the source and adjacent to the skull.

In total, these experiments illustrate that larger patch sizes make the approach more dependent on the global mesh structure, while small patch sizes mainly depend on the local mesh structure in the proximity of the source. In particular, this is one of the main advantages of the localized subtraction approach over the analytical subtraction approach.

Importantly, an increase in vertex extensions is always accompanied by an increase in computational costs. To achieve the best speed and accuracy, we thus suggest using two vertex extensions for the patch construction.

\subsection{Comparison to other potential approaches}
We evaluate the localized subtraction approach against the analytical subtraction approach \cite{beltrachini_analytic_subtraction}. Furthermore, we want to test the localized subtraction approach against the multipolar Venant approach \cite{vorwerk_venant}, as it is a widely used and well-performing direct approach. In the following, we abbreviate the localized subtraction approach as ''loc. sub.'', the analytical subtraction approach as ''ana. sub.'' and the multipolar Venant approach as ''mul. Venant''.

\subsubsection{Accuracy comparison}
We first want to compare the accuracies of the different potential approaches. For this, we computed forward solutions for radial and tangential dipoles at different, logarithmically scaled eccentricities, ranging from 0.8803 to 0.9900. At each eccentricity, we used 1000 approximately uniformly distributed dipoles and computed the relative error with respect to the analytic solutions. For the EEG forward simulation, we shifted all forward solutions to have zero mean. For the MEG forward simulation, we computed the whole magnetic field vector. Furthermore, we only solved the MEG forward problem for tangential dipoles, since radial dipoles produce no magnetic field outside of a spherically symmetric volume conductor \cite{sarvas_analytical_meg}. The results for \textit{mesh\_init} (see Figure \ref{meshes_refinement}\,(a)) are shown in Figure \ref{source_model_accuracy_comparison_eeg_radial} and Figure \ref{source_model_accuracy_comparison_meg_tangential}.

\begin{figure}
\centerline{\includegraphics{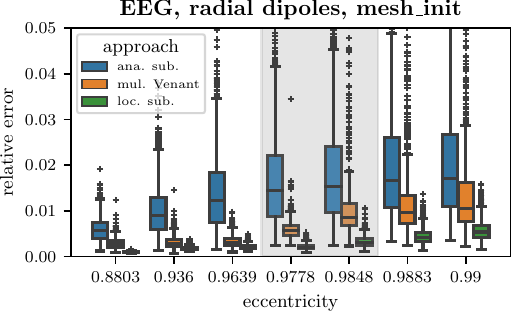}}
\caption{Accuracy comparison of EEG forward simulations using the analytical subtraction, multipolar Venant, and localized subtraction potential approaches for radial dipoles at different eccentricities using \textit{mesh\_init} (see Figure \ref{meshes_refinement}\,(a)). The $y$-axis shows the relative error. The physiologically relevant sources at 1-2 mm distance from the CSF are highlighted.} 
\label{source_model_accuracy_comparison_eeg_radial}
\end{figure}

\begin{figure}
\centerline{\includegraphics{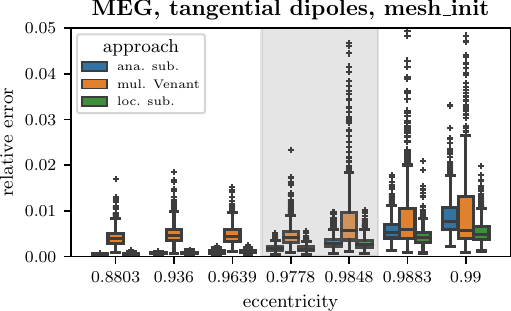}}
\caption{Accuracy comparison of MEG forward simulations using the analytical subtraction, multipolar Venant, and localized subtraction potential approaches for tangential dipoles at different eccentricities using \textit{mesh\_init} (see Figure \ref{meshes_refinement}\,(a)). The $y$-axis shows the relative error. The physiologically relevant sources at 1-2 mm distance from the CSF are highlighted.} 
\label{source_model_accuracy_comparison_meg_tangential}
\end{figure}

We first discuss the EEG case. For radial sources, we see in Figure \ref{source_model_accuracy_comparison_eeg_radial} that in \textit{mesh\_init} the localized subtraction approach outperforms the analytical subtraction and multipolar Venant approaches over all investigated eccentricities and produces a median error of below $1\%$ even for sources at $0.99$ eccentricity (meaning a distance of $0.8$ mm to the conductivity jump). Note that we computed the right-hand sides for the analytical subtraction approach using the analytical formulas from \cite{beltrachini_analytic_subtraction}, their larger errors are thus not due to an inaccurate integration but instead result from a relatively low resolution in the skin compartment, as described in section \ref{patch_construction_subsection}. The EEG results for tangential dipoles show similar behavior, with the errors being overall smaller (see Figure \ref{accuracy_comparison_eeg_tangential_mesh_init}).

Performing the same comparison in \textit{mesh\_brain} shows that, when compared to the results in Figure \ref{source_model_accuracy_comparison_eeg_radial}, refining the brain leads to a large decrease in the errors of the multipolar Venant approach, a small decrease for the localized subtraction approach and a barely noticeable decrease for the analytical subtraction approach (see Figures \ref{accuracy_comparison_eeg_radial_mesh_brain} and \ref{accuracy_comparison_eeg_tangential_mesh_brain} for the results and Figure \ref{meshes_refinement}\,(b) for the mesh). Furthermore, in this mesh, the localized subtraction approach and the multipolar Venant approach show approximately the same accuracy. Comparing Figure \ref{source_model_accuracy_comparison_eeg_radial} against the corresponding results obtained using \textit{mesh\_skin}, we instead see that the errors for the multipolar Venant approach and the localized subtraction approach barely change, while the errors for the analytical subtraction approach notably decrease (see Figures \ref{accuracy_comparison_eeg_radial_mesh_skin} and \ref{accuracy_comparison_eeg_tangential_mesh_skin} for the results and Figure \ref{meshes_refinement}\,(c) for the mesh). In particular, in this mesh, the analytical subtraction approach even surpasses the multipolar Venant approach at some eccentricities with regard to accuracy but does not reach the accuracy of the localized subtraction approach at any eccentricity.

These results illustrate the limitations of the different potential approaches. Refining the brain compartment enables a more favorable placement of monopolar loads, which leads to a decrease in errors for the multipolar Venant approach when going from \textit{mesh\_init} to \textit{mesh\_brain}. Since the resolution in the brain compartment is the limiting factor for the placement of monopolar loads, this also explains why the error for the multipole approach hardly changes when refining the skin compartment. The analytical subtraction approach on the other hand shows the opposite behavior. In this case, the low resolution of the skull-skin interface in \textit{mesh\_init} is the limiting factor, and hence we see a large error decrease after refining the skin, and only a very minor decrease after refining the brain. Similar to the analytical subtraction approach, the localized subtraction approach depends on a high resolution in the region where the correction potential shows a complicated behavior. But in contrast to the analytical subtraction approach, this region is confined to the close proximity of the source positions, and hence the high CSF resolution in \textit{mesh\_init} already suffices to achieve a good approximation of $u^c = u - \chi \cdot u^\infty$. In particular, this explains why the errors for the localized subtraction approach only slightly decrease when refining either the brain or the skin.

We now discuss the MEG case. Figure \ref{source_model_accuracy_comparison_meg_tangential} shows the accuracy comparisons for \textit{mesh\_init}. We see that, in this mesh, both the analytical subtraction and localized subtraction approach produce highly accurate results and are in particular more accurate than the multipolar Venant approach in the region of 1-2 mm distance to the conductivity jump. It is especially noteworthy that the analytical subtraction approach is able to produce such accurate results, as it produced larger errors in the EEG forward problem in this mesh. The reason for the EEG errors not having a noticeable influence on the MEG accuracy is that the errors in the EEG forward problem for the analytical subtraction approach are mainly located in the skin compartment.  But since the skin compartment is distant from the source, with the low-conducting skull in between,  the volume currents in the skin are low in magnitude and hence their contribution to the magnetic field is relatively small. In particular, these results demonstrate the strong influence of the volume currents in the close proximity of the source on the magnetic field. The accuracy of the analytical subtraction and localized subtraction MEG forward solutions also shows the impact of faithful current modeling, which was already suggested by Figure \ref{flux_loc_sub}. Furthermore, in contrast to the EEG case, the larger errors of the analytical subtraction approach for the higher eccentricities are entirely due to the numerical quadrature. When choosing a sufficiently large integration order for both approaches, the MEG analytical subtraction approach\footnote{
Note that the \textit{analytical} subtraction approaches \textit{analytically} computes the EEG right hand sides. If one then uses the Biot-Savart law to compute the magnetic field, a numerical integration is still required.}
is as accurate as the localized subtraction approach. In Figure \ref{source_model_accuracy_comparison_meg_tangential}, for the analytical subtraction approach an integration order of $5$ was chosen to achieve a somewhat reasonable computation time, while the integration orders for the localized subtraction approach were chosen as described in \ref{integration_order_chapter}.

When performing the same comparison for \textit{mesh\_skin}, we see that the errors are essentially the same as in Figure \ref{source_model_accuracy_comparison_meg_tangential} (see Figure \ref{accuracy_comparison_meg_tangential_mesh_skin}). If we furthermore look at a single potential approach and compare the errors in \textit{mesh\_init} and \textit{mesh\_brain}, we see that for all three investigated approaches the errors were smaller in \textit{mesh\_brain}. Moreover, in \textit{mesh\_brain} the different potential approaches are comparable with regard to accuracy (see Figure \ref{accuracy_comparison_meg_tangential_mesh_brain}).

The results in \textit{mesh\_skin} and \textit{mesh\_brain} again highlight the dominant role of the volume currents in the close vicinity of the source. Refining the skin does not enable a better representation of the volume currents close to the sources while refining the brain does. This explains why the errors do not change when going from \textit{mesh\_init} to \textit{mesh\_skin}, but decrease when going from \textit{mesh\_init} to \textit{mesh\_brain}.

Furthermore, these results are in line with previous investigations, where it was shown that the MEG forward problem is most sensitive to conductivity uncertainties in the proximity of the dipole location \cite{gencer_transfer_matrices}, and is insensitive to the conductivities of the skull and the skin \cite{handbook_of_neural_activity_measurement}.

\subsubsection{Efficiency comparison}
Finally, we now want to compare the computational effort of the different potential approaches. For this we measured the time it took to solve the EEG and MEG forward problems for $1000$ right-hand sides for sources at an eccentricity of $0.99$, assuming the transfer matrix has already been computed. Hence, we are measuring the time it takes to assemble the right-hand sides and then compute the product of these right-hand sides with the transfer matrix. We chose an eccentricity of $0.99$ since the effort needed for the localized subtraction right-hand side assembly depends on the distance of the source to the conductivity jump. One needs to assemble the patch, the boundary, and the transition integrals and fluxes as described in \eqref{transition_integral_eeg_main_paper} - \eqref{patch_integral_eeg_main_paper} and \eqref{boundary_subtraction}. The computational effort for the surface and transition integrals does not depend on the eccentricity, but the effort of computing the patch integrals increases for sources closer to the conductivity jump. More concretely, the patch integrals vanish for every element in the same compartment as the source, as there we have $\sigma^c = 0$, and hence we see a larger number of non-zero patch integrals when increasing the eccentricity, as a larger portion of the patch grows outside the brain compartment. Furthermore, for highly eccentric sources the patch elements with a non-zero integral get close to the source, and hence we need a high integration order on these elements.\footnote{We refer to \ref{integration_order_chapter} for an in-depth discussion of integration orders.} By measuring the computation time for sources with an eccentricity of $0.99$, which is slightly above the physiologically relevant range, we thus arrive at a worst-case estimate for the runtime of the localized subtraction approach. The experiments were run on an AMD Ryzen Threadripper 3960X CPU. The results for the different meshes (see Figure \ref{meshes_refinement}) are shown in  the tables \ref{time_comparison_mesh_init} - \ref{time_comparison_mesh_skin}. We see that the time needed for the multipole approach and for the EEG localized subtraction approach are essentially independent of the mesh size\footnote{Note that the mesh size of course still influences the time needed for the computation of the transfer matrices.}. Furthermore, we see that the computation time for the MEG localized subtraction approach is the smallest for \textit{mesh\_brain}. This can be explained by the fact that a higher resolution in the brain compartment leads to a larger portion of patch elements lying in the same compartment as the source, which eliminates the need for an expensive numerical quadrature on these elements. And most importantly, we see that in all meshes the localized subtraction approach is orders of magnitude faster than the analytical subtraction approach. Hence the localized subtraction approach achieves its central goal of being a computationally efficient approach based on subtraction of the singularity. Furthermore, note that the time needed for the computation of the transfer matrices is considerably larger than the time needed for the subsequent localized subtraction right-hand side assembly. Since for state-of-the-art direct potential approaches, such as the multipolar Venant approach \cite{vorwerk_venant}, the total time needed to solve the forward problems is also dominated by the time needed to compute the transfer matrices, the localized subtraction approach is thus a viable alternative to direct FEM potential approaches in practical applications.

\begin{table}
\caption{Forward simulation times for $1000$ sources at an eccentricity of $0.99$ in \textit{mesh\_init}}
\label{time_comparison_mesh_init}
\begin{tabular}{|p{55pt}|p{50.65pt}|p{50.65pt}|p{50.65pt}|}
\hline
\small Approach & \small loc. sub.&\small ana. sub. & \small mul. Venant\\ \hline
{\small EEG time (s) }& \small$0.93$ & \small $655$ & \small $0.02$ \\ \hline
\small MEG time (s) & \small $13$ & \small $18925$ & \small $0.04$ \\ \hline\end{tabular}
\end{table}

\begin{table}
\caption{Forward simulation times for $1000$ sources at an eccentricity of $0.99$ in \textit{mesh\_brain}}
\label{time_comparison_mesh_brain}
\begin{tabular}{|p{55pt}|p{50.65pt}|p{50.65pt}|p{50.65pt}|}
\hline
\small Approach & \small loc. sub.& \small ana. sub. & \small mul. Venant\\ \hline
\small EEG time (s) & \small $0.96$ & \small $1396$ & \small $0.02$ \\ \hline
\small MEG time (s) & \small $6.8$ & \small $43175$ & \small $0.04$ \\ \hline\end{tabular}
\end{table}

\begin{table}
\caption{Forward simulation times for $1000$ sources at an eccentricity of $0.99$ in \textit{mesh\_skin}}
\label{time_comparison_mesh_skin}
\begin{tabular}{|p{55pt}|p{50.65pt}|p{50.65pt}|p{50.65pt}|}
\hline
\small Approach & \small loc. sub.& \small ana. sub. & \small mul. Venant\\ \hline
\small EEG time (s) & \small $0.98$ & \small $1142$ & \small $0.02$ \\ \hline
\small MEG time (s) & \small $12.8$ & \small $30709$ & \small $0.04$ \\ \hline\end{tabular}
\end{table}

\section{Conclusion}
We presented a new potential approach based on a localization of the well-known subtraction approach. We investigated the theory of this new approach and validated and compared it against established approaches. In spherical head models, the new approach was shown to rival, and in some cases even surpass, existing approaches in terms of accuracy while being largely more efficient than the subtraction approach. We believe the localized subtraction to be the superior way of analytically dealing with the singularity of the mathematical point dipole, and a competitive alternative in practical applications.

%%%%%%%%%%%%%%%%%%%%%%%%%%%%%%%%%%%%%%%%%%
% Formal requirements from neuroimage
%%%%%%%%%%%%%%%%%%%%%%%%%%%%%%%%%%%%%%%%%%
\section*{Data/code availability statement}
The code implementing the presented approach is available at \url{https://gitlab.dune-project.org/duneuro/duneuro/-/tree/feature/localized-assembler}. The code used to test different integration orders is available at \url{https://github.com/MalteHoel/local_integration_test}. The data used in the experiments is available at \url{https://doi.org/10.5281/zenodo.7642826}.

\section*{Declaration of competing interests}
Declarations of interest: none.

\section*{Credit authorship contribution statement}
\begin{sloppypar}
\textbf{Malte Höltershinken:} Methodology, Software, Formal analysis, Investigation, Writing - original draft, Writing - review and editing, Visualization.
\textbf{Pia Lange:} Methodology, Software, Investigation, Writing - review and editing.
\textbf{Tim Erdbrügger:} Writing - review and editing.
\textbf{Yvonne Buschermöhle:} Writing - review and editing.
\textbf{Fabrice Wallois:} Writing - review and editing, Funding Acquisition.
\textbf{Alena Buyx:} Funding Acquisition.
\textbf{Sampsa Pursiainen:} Writing - review and editing, Funding Acquisition.
\textbf{Johannes Vorwerk:} Investigation, Writing - review and editing, Funding Acquisition.
\textbf{Christian Engwer:} Conceptualization, Methodology, Software, Writing - review and editing, Supervision.
\textbf{Carsten Wolters:} Conceptualization, Methodology, Writing - review and editing, Supervision, Funding Acquisition.
\end{sloppypar}

\section*{Acknowledgements}
\begin{sloppypar}
This work was supported by the Bundesministerium für Gesundheit (BMG) as project 
\mbox{ZMI1-2521FSB006},
under the frame of ERA PerMed as project 
\mbox{ERAPERMED2020-227} 
(MBH, PL, TE, YB, FW, AB, CHW), by the Deutsche Forschungsgemeinschaft  (DFG), project 
\mbox{WO1425/10-1}
(MBH, TE, YB, CHW), by the Deutsche Forschungsgemeinschaft through the Cluster of Excellence “Mathematics Münster: Dynamics - Geometry - Structure”
(\mbox{EXC 2044-390685587})
(CE), by the Academy of Finland
(\mbox{decision 344712})
under the frame of ERA PerMed as project 
\mbox{ERAPERMED2020-227} 
(SP), by the AoF Center of Excellence in Inverse Modelling and Imaging 2018-2025
(\mbox{decision 353089})
(SP), by the Austrian Wissenschaftsfonds (FWF), project 
\mbox{P 35949}
(JV), and by DAAD project 
\mbox{57663920}
(MBH, TE, SP, CHW). We acknowledge support from the Open Access Publication Fund of the University of Muenster.
\end{sloppypar}

%%%%%%%%%%%%%%%%%%%%%%%%%%%%%%%%%%%%%%%%%%
%%%%%%%%%%%%%%%%%%%%%%%%%%%%%%%%%%%%%%%%%%
%%%%%%%%%%%%%%%%%%%%%%%%%%%%%%%%%%%%%%%%%%
% Appendix
%%%%%%%%%%%%%%%%%%%%%%%%%%%%%%%%%%%%%%%%%%
%%%%%%%%%%%%%%%%%%%%%%%%%%%%%%%%%%%%%%%%%%
%%%%%%%%%%%%%%%%%%%%%%%%%%%%%%%%%%%%%%%%%%
\appendix
\renewcommand{\thelemma}{\Alph{section}\arabic{lemma}}

\section{Proof of lemma 1}
\label{proof_of_distributional_derivative}
\setcounter{figure}{0}
\setcounter{table}{0}
\setcounter{lemma}{0}
In the following, we assume $\Omega \subset \mathbb{R}^3$ to be open. Furthermore, we denote by $C^\infty_c(\Omega)$ the space of all infinitely differentiable, compactly supported functions on $\Omega$. Note that compactly supported functions on $\Omega$ necessarily vanish near the boundary of $\Omega$, which implies that boundary integrals involving such functions vanish. This in particular applies to the boundary integrals arising during partial integration.

To give a mathematically rigorous derivation of the method, we will make use of distributions. A \textit{distribution} (or \textit{generalized function}) is defined as a linear functional $T: C^\infty_c(\Omega) \rightarrow \mathbb{R}$ which is continuous with respect to a certain topology. Every locally integrable function $f$ defines a distribution $T_f$ via
\[
T_f(\phi) = \int_{\Omega} f \phi \, dV.
\]
In the following, we will simply write $T_f(\phi) = f(\phi)$. A central property of distributions is that they possess arbitrary derivatives, while ordinary functions may fail to be (weakly) differentiable. It is thus natural to state and investigate partial differential equations in a distributional framework. Nevertheless, we want to emphasize that the following derivation of the localized subtraction approach does not assume the reader to be familiar with this theory. For readers interested in an in-depth introduction to distributions and their application to differential equations, we refer to \cite{rudin_functional_analysis}, chapters 6-8.

For reference, we will first recall lemma \ref{distributional_derivative_u_infty_times_chi}.
\begin{lemma*}
Let $\Omega^\infty \subset \Omega$ and $\chi \in H^1(\Omega)$ be such that $\chi$ is identical to $1$ on $\Omega^\infty$ and $x_0 \in \Omega^\infty$. Let $\widetilde{\Omega} \subset \Omega\setminus\Omega^\infty$ be a region such that $\partial\Omega^\infty$ is a subset of $\partial\Omega \cup \partial\widetilde{\Omega}$, and $\chi = 0$ outside $\Omega^\infty \cup \widetilde{\Omega}$. Then for any test function $\phi \in C^\infty_c(\Omega)$ we have
\begin{multline}
\divergence\left(\sigma^\infty \nabla \left(\chi \cdot u^\infty \right)\right)(\phi) 
 = \\ \divergence\left(M \cdot \delta_{x_0}\right)(\phi)
 - \int_{\widetilde{\Omega}} \langle\sigma^\infty \nabla \left(\chi \cdot u^\infty \right), \nabla \phi \rangle\, dV \\
 - \int_{\partial\Omega^\infty} \langle\sigma^\infty \nabla u^\infty , \eta\rangle \phi \, dS.
\label{distributional_derivative_chi_u_inf}
\end{multline}
Here $\eta$ is the unit outer normal of $\Omega^\infty$.
\end{lemma*}

\begin{proof}
First, note that $\chi \cdot u^\infty$ is locally integrable, and hence defines a distribution. We then have by definition
\begin{align*}
\divergence(\sigma^\infty \nabla(\chi \cdot u^\infty)) (\phi)
\stackrel{\footnotemark}{=} \int_{\Omega} \chi \cdot u^\infty \divergence(\sigma^\infty \nabla \phi) \, dV \\
= \int_{\Omega} u^\infty \divergence(\sigma^\infty \nabla \phi) \, dV
+ \int_{\Omega} (\chi - 1) u^\infty \divergence(\sigma^\infty \nabla \phi) \, dV \\
= \divergence(\sigma^\infty \nabla u^\infty)(\phi) 
+ \int_{\Omega\setminus\Omega^\infty} (\chi - 1) u^\infty \divergence(\sigma^\infty \nabla \phi) \, dV \\
= \divergence(M \cdot \delta_{x_0})(\phi) 
+ \int_{\Omega\setminus\Omega^\infty} (\chi - 1) u^\infty \divergence(\sigma^\infty \nabla \phi) \, dV.
\end{align*}
\footnotetext{If we interpret the left-hand side in the distributional sense, the right-hand side of this equation is essentially the definition of the left-hand side. In a less rigorous manner, one might write the left-hand side as 
\[
\int_{\Omega}\divergence(\sigma^\infty \nabla(\chi \cdot u^\infty)) \phi \, dV,
\]
apply partial integration two times and note that the boundary integrals vanish since $\phi = 0$ and $\nabla\phi = 0$ on the boundary of $\Omega$. Note that this is indeed ''less rigorous'' since e.g. $\nabla(\chi \cdot u^\infty)$ is not integrable. This illustrates the benefit of interpreting the derivatives in the distributional sense.
}
Here we have used for the first summand that $u^\infty$ solves $\divergence(\sigma^\infty \nabla u^\infty) = \divergence(M \cdot \delta_{x_0})$ on $\mathbb{R}^3$ and that $\phi$ has compact support in $\Omega$. For the second summand, we have used $\chi = 1$ on $\Omega^\infty$. Note that $(\chi - 1) u^\infty$ is weakly differentiable on $\Omega\setminus\Omega^\infty$, and hence we can apply partial integration to the second summand. Letting $\eta$ denote the corresponding unit outer normal, this yields
\begin{multline*}
\int_{\Omega\setminus\Omega^\infty} (\chi - 1) u^\infty \divergence(\sigma^\infty \nabla \phi) \, dV \\
= \int_{\partial(\Omega \setminus \Omega^\infty)} (\chi - 1) u^\infty \langle \sigma^\infty\nabla \phi, \eta\rangle\, dS\\
- \int_{\Omega\setminus\Omega^\infty} \langle\sigma^\infty\nabla((\chi - 1) u^\infty) , \nabla \phi\rangle \, dV
\end{multline*}
Since $\partial(\Omega\setminus\Omega^\infty) \subset \partial \Omega \cup \partial \Omega^\infty$ and we have $\nabla \phi = 0$ on $\partial \Omega$ and $\chi - 1 = 0$ on $\partial\Omega^\infty$ the boundary integral vanishes. For the volume integral, we get by again applying partial integration
\begin{multline*}
- \int_{\Omega\setminus\Omega^\infty} \langle\sigma^\infty\nabla((\chi - 1) u^\infty) , \nabla \phi\rangle \, dV\\
= -\int_{\Omega\setminus\Omega^\infty} \langle \sigma^\infty \nabla(\chi \cdot u^\infty), \nabla \phi\rangle \, dV
+ \int_{\Omega\setminus\Omega^\infty} \langle \sigma^\infty \nabla u^\infty, \nabla\phi\rangle\, dV \\
= -\int_{\Omega\setminus\Omega^\infty} \langle \sigma^\infty \nabla(\chi \cdot u^\infty), \nabla \phi\rangle \, dV
+ \int_{\partial(\Omega\setminus\Omega^\infty)} \langle \sigma^\infty \nabla u^\infty, \eta\rangle \phi \, dS\\
- \int_{\Omega\setminus\Omega^\infty} \divergence(\sigma^\infty \nabla u^\infty) \phi \, dV.
\end{multline*}
Since on $\mathbb{R}^3\setminus\{x_0\}$ we have $\divergence(\sigma^\infty \nabla u^\infty) = 0$, the last integral vanishes. For the boundary integral, since $\phi$ has compact support in $\Omega$ the portion of the integral on the boundary of $\Omega$ vanishes. Noting that the outer normal of $\Omega\setminus\Omega^\infty$ on the boundary of $\Omega^\infty$ in the interior of $\Omega$ is the negative of the outer normal of $\Omega^\infty$, we arrive at \eqref{distributional_derivative_chi_u_inf}. We want to emphasize that this derivation also covers the case where the patch touches the boundary of $\Omega$ since in this case, the corresponding portion of the boundary integral in \eqref{distributional_derivative_chi_u_inf} vanishes since $\phi$ has compact support. This completes the proof.
\end{proof}

\section{Derivation of the weak formulation}
\label{weak_formulation_correct}
\setcounter{figure}{0}
\setcounter{table}{0}
\setcounter{lemma}{0}
Let
\begin{equation}
a(w, v) = \int_{\Omega} \langle \sigma \nabla w, \nabla v \rangle\, dV
\end{equation}
and 
\begin{multline}
l(v) = - \int_{\widetilde{\Omega}} \langle \sigma \nabla \left(\chi \cdot u^\infty \right), \nabla v\rangle \, dV \\
- \int_{\partial \Omega^\infty} \langle\sigma^\infty \nabla u^\infty, \eta \rangle v \, dS
- \int_{\Omega^\infty} \langle \sigma^c \nabla u^\infty, \nabla v \rangle \, dV
\label{localized_subtraction_rhs}
\end{multline}
be defined as in the main paper. Then $a$ defines a continuous bilinear form and $l$ defines a continuous linear form on $H^1(\Omega)$. We want to show that, for $u^c \in H^1(\Omega)$ a solution of $a(u^c, v) = l(v)$ for all $v \in H^1(\Omega)$, the function $u^c + \chi \cdot u^\infty$ solves the original equation in a suitable sense. Reading this backward can serve as a derivation of the weak formulation.

\begin{lemma}
Let $u^c \in H^1(\Omega)$ be a solution of $a(u^c, v) = l(v)$ for all $v \in H^1(\Omega)$. Then $u = u^c + \chi \cdot u^\infty$ solves
\begin{align*}
\divergence(\sigma \nabla u) &= \divergence(M \cdot \delta_{x_0}) &\text{ on $\Omega$} \\
\langle \sigma \nabla u, \eta\rangle &= 0 &\text{ on $\partial\Omega$}.
\end{align*}
\end{lemma}

\medskip

\begin{proof}
Since $u^c \in H^1(\Omega)$, $u^\infty$ is locally integrable and $\chi$ is bounded, it follows that $u$ is locally integrable and hence defines a distribution. We then have for a test function $\phi \in  C^\infty_c(\Omega)$
\begin{multline*}
\divergence(\sigma \nabla u)(\phi)
= \divergence(\sigma \nabla u^c)(\phi) 
+ \divergence(\sigma^c \nabla(\chi\cdot u^\infty))(\phi) \\
+ \divergence(\sigma^\infty \nabla(\chi \cdot u^\infty))(\phi) \\
= -\int_{\Omega} \langle \sigma \nabla u^c, \nabla \phi\rangle \, dV
- \int_{\Omega} \langle \sigma^c \nabla (\chi \cdot u^\infty), \nabla \phi\rangle \, dV \\
+ \divergence(M \cdot \delta_{x_0})(\phi)
- \int_{\widetilde{\Omega}} \langle \sigma^\infty \nabla(\chi \cdot u^\infty), \nabla \phi\rangle \, dV \\
- \int_{\partial\Omega^\infty} \langle\sigma^\infty \nabla u^\infty , \eta\rangle \phi \, dS \\
= \divergence(M \cdot \delta_{x_0})(\phi) - a(u^c, \phi) 
 - \int_{\widetilde{\Omega}} \langle \sigma \nabla \left(\chi \cdot u^\infty \right), \nabla \phi\rangle \, dV \\
- \int_{\partial \Omega^\infty} \langle\sigma^\infty \nabla u^\infty, \eta \rangle \phi \, dS
- \int_{\Omega^\infty} \langle \sigma^c \nabla u^\infty, \nabla \phi \rangle \, dV \\
= \divergence(M \cdot \delta_{x_0})(\phi),
\end{multline*}
where the last equality follows from $a(u^c, \phi) = l(\phi)$ and the definition of $l$. We have thus shown that $u$ solves
\begin{align*}
\divergence(\sigma \nabla u) &= \divergence(M \cdot \delta_{x_0}) &\text{ on $\Omega$}.
\end{align*}
Now note that since $\divergence(\sigma^\infty \nabla u^\infty) = 0$ on $\mathbb{R}^3\setminus\{x_0\}$, we have for a smooth function $\phi \in C^\infty\left(\overline{\Omega}\right)$ with $\support(\phi)\cap \{x_0\} = \emptyset$ that
\begin{equation*}
\int_{\partial \Omega^\infty} \langle\sigma^\infty \nabla u^\infty, \eta \rangle \phi \, dS
= \int_{\Omega^\infty} \langle \sigma^\infty \nabla u^\infty, \nabla \phi\rangle \, dV.
\end{equation*}
We thus have
\begin{multline*}
l(\phi) = - \int_{\widetilde{\Omega}} \langle \sigma \nabla \left(\chi \cdot u^\infty \right), \nabla \phi\rangle \, dV \\
- \int_{\partial \Omega^\infty} \langle\sigma^\infty \nabla u^\infty, \eta \rangle \phi \, dS
- \int_{\Omega^\infty} \langle \sigma^c \nabla u^\infty, \nabla \phi \rangle \, dV\\
= - \int_{\widetilde{\Omega}} \langle \sigma \nabla \left(\chi \cdot u^\infty \right), \nabla \phi\rangle \, dV
- \int_{\Omega^\infty} \langle \sigma \nabla u^\infty, \nabla \phi\rangle\, dV\\
=  - \int_{\widetilde{\Omega}} \langle \sigma \nabla \left(\chi \cdot u^\infty \right), \nabla \phi\rangle \, dV
-  \int_{\Omega^\infty} \langle \sigma \nabla \left(\chi \cdot u^\infty \right), \nabla \phi\rangle \, dV\\
= -\int_{\Omega}\langle \sigma \nabla \left(\chi \cdot u^\infty \right), \nabla \phi\rangle \, dV.
\end{multline*}
Inserting this into $a(u^c, \phi) - l(\phi) = 0$ yields
\begin{equation*}
\int_{\Omega} \langle \sigma \nabla u, \nabla \phi\rangle \, dV = 0.
\end{equation*}
Since $\divergence(\sigma\nabla u) = 0$ on $\Omega\setminus\{x_0\}$ applying partial integration again yields
\begin{equation*}
\int_{\partial \Omega} \langle \sigma\nabla u, \eta\rangle \phi \, dS = 0,
\end{equation*}
which yields $\langle \sigma \nabla u, \eta\rangle = 0$ on $\partial\Omega$. This completes the proof.
\end{proof}

\section{Well-posedness of the problem}
\label{well_defined_proof}
\setcounter{figure}{0}
\setcounter{table}{0}
\setcounter{lemma}{0}
In the main paper, we claimed that the functional $l$ vanishes on constant functions and that the family of solutions to the total potential $u$ does not depend on the choice of $\chi, \Omega^\infty$, and $\widetilde{\Omega}$. We want to give proof of these claims.

\begin{lemma}
\label{functional_vanishes_on_constants}
The functional $l$ in \eqref{localized_subtraction_rhs} vanishes on constant functions.
\end{lemma}

\begin{proof}
Using $\nabla 1 = 0$ we get
\begin{equation*}
l(1) = 0 \iff \int_{\partial\Omega^\infty} \langle \sigma^\infty \nabla u^\infty, \eta\rangle \, dS = 0.
\end{equation*}
Now let $R > 0$ be large enough so that $\Omega^\infty$ is contained inside the interior of the sphere $S_R(x_0)$ of radius $R$ around $x_0$. Let $W \subset \mathbb{R}^3$ denote region bounded by $\Omega^\infty$ and $S_R(x_0)$. Since $\divergence(\sigma^\infty \nabla u^\infty) = 0$ on $\mathbb{R}^3\setminus\{x_0\}$ the divergence theorem then implies
\begin{equation*}
0 = \int_W \divergence\left(\sigma^\infty \nabla u^\infty\right)\, dV
= \int_{\partial W} \langle \sigma^\infty \nabla u^\infty, \eta \rangle\, dS,
\end{equation*}
where $\eta$ denotes the outer normal of $W$. Since we have $\partial W$ = $S_R(x_0) \cup \partial \Omega^\infty$, and the outer normal of $W$ on $\partial\Omega^\infty$ is the negative of the outer normal of $\Omega^\infty$, this implies
\begin{equation*}
\int_{\partial\Omega^\infty} \langle \sigma^\infty \nabla u^\infty, \eta\rangle \, dS
= \int_{S_R(x_0)} \langle \sigma^\infty \nabla u^\infty, \eta\rangle \, dS,
\end{equation*}
where $\eta$ again denotes the corresponding unit outer normals.
We now have for $x \in \mathbb{R}^3\setminus\{x_0\}$ that
\begin{align*}
u^\infty(x) =
\frac{1}{4 \pi \sqrt{\det\sigma^\infty}} \cdot \frac{\langle M, \left(\sigma^\infty\right)^{-1} (x - x_0)\rangle}{\|\left(\sigma^\infty\right)^{-\frac{1}{2}}(x - x_0)\|^3},
\end{align*}
see e.g. \cite{wolters_subtraction_method}. Denoting $R= \|x - x_0\|$, a straightforward computation shows that $\langle \sigma^\infty \nabla u^\infty, \eta\rangle \in \mathcal{O}\left(\frac{1}{R^3}\right)$. This implies 
\begin{equation*}
\int_{S_R(x_0)} \langle \sigma^\infty \nabla u^\infty, \eta\rangle \, dS \in \mathcal{O}\left(\frac{1}{R}\right),
\end{equation*}
and hence
\begin{equation*}
\int_{\partial\Omega^\infty} \langle \sigma^\infty \nabla u^\infty, \eta\rangle \, dS
= \int_{S_R(x_0)} \langle \sigma^\infty \nabla u^\infty, \eta\rangle \, dS
\stackrel{R \rightarrow \infty}{\longrightarrow} 0.
\end{equation*}
This is only possible if
\begin{equation*}
\int_{\partial\Omega^\infty} \langle \sigma^\infty \nabla u^\infty, \eta\rangle \, dS = 0.
\end{equation*}
\end{proof}

Note that the previous lemma essentially consists of showing that 
\[
\int_{\partial\Omega^\infty} \langle \sigma^\infty \nabla u^\infty, \eta\rangle \, dS = 0.
\]
The same integral also appears in \cite{wolters_subtraction_method}, where the authors use a symmetry argument to deduce that it is equal to $0$. Since our argument in lemma \ref{functional_vanishes_on_constants} takes a somewhat different approach, we deemed it interesting enough to include.

\begin{lemma}
Let $\Omega^\infty_1, \widetilde{\Omega}_1, \chi_1$ and $\Omega^\infty_2, \widetilde{\Omega}_2, \chi_2$ be two choices of patch, transition region, and cutoff function. Let $u^c_1$ and $u^c_2$ be solutions to the corresponding weak formulations. Then we have
\[
u^c_1 + \chi_1 \cdot u^\infty + \mathbb{R} \cdot 1 = u^c_2 + \chi_2 \cdot u^\infty + \mathbb{R} \cdot 1.
\]
The potentials defined by the localized subtraction approach thus do not depend on the choice of the patch, transition region, and cutoff function.
\end{lemma}

\begin{proof}
It obviously suffices to show that any localized subtraction approach produces the same solution family as the ordinary subtraction approach. We can thus without loss of generality assume $\Omega^\infty_2 = \Omega$, $\widetilde{\Omega}_2 = \emptyset$ and $\chi_2 = 1$. 

Now note that
\[
u^c_1 - u^c_2 + (\chi_1 - 1) \cdot u^\infty \in H^1(\Omega),
\]
since $\chi_1 - 1$ vanishes on an environment of $x_0$. Then we have for $v \in H^1(\Omega)$
\begin{multline}
a(u^c_1 - u^c_2 + (\chi_1 - 1) \cdot u^\infty, v) \\
= a(u^c_1, v) - a(u^c_2, v) + \int_{\Omega} \langle \sigma \nabla(\chi_1 - 1) u^\infty, \nabla v\rangle \, dV.
\label{diff_weak_solutions_loc_sub_sub}
\end{multline}
Since $u^c_1$ and $u^c_2$ are solutions to their respective weak problems, we get
\begin{multline*}
a(u^c_1, v) - a(u^c_2, v)
=  -\int_{\widetilde{\Omega}} \langle \sigma \nabla \left(\chi_1 \cdot u^\infty \right), \nabla v\rangle \, dV \\
- \int_{\partial \Omega^\infty} \langle\sigma^\infty \nabla u^\infty, \eta \rangle v \, dS
- \int_{\Omega^\infty} \langle \sigma^c \nabla u^\infty, \nabla v \rangle \, dV \\
+ \int_{\Omega} \langle \sigma^c \nabla u^\infty, \nabla v\rangle \, dV
+ \int_{\partial \Omega} \langle \sigma^\infty \nabla u^\infty, \eta\rangle v \, dS \\
= \left(\int_{\partial \Omega} \langle \sigma^\infty \nabla u^\infty, \eta\rangle v \, dS
-  \int_{\partial \Omega^\infty} \langle\sigma^\infty \nabla u^\infty, \eta \rangle v \, dS \right) \\
+ \int_{\Omega\setminus\Omega^\infty} \langle \sigma^c \nabla u^\infty, \nabla v\rangle \, dV
- \int_{\widetilde{\Omega}} \langle \sigma \nabla \left(\chi_1 \cdot u^\infty \right), \nabla v\rangle \, dV \\
= \int_{\Omega\setminus\Omega^\infty} \langle \sigma^\infty \nabla u^\infty, \nabla v\rangle \, dV 
+ \int_{\Omega\setminus\Omega^\infty} \langle \sigma^c \nabla u^\infty, \nabla v\rangle \, dV \\
- \int_{\widetilde{\Omega}} \langle \sigma \nabla \left(\chi_1 \cdot u^\infty \right), \nabla v\rangle \, dV \\
= \int_{\Omega\setminus\Omega^\infty} \langle \sigma \nabla \left(1 - \chi_1) u^\infty\right) , \nabla v\rangle \, dV.
\end{multline*}
Inserting this into \eqref{diff_weak_solutions_loc_sub_sub} and noting that $\chi_1 - 1 = 0$ on $\Omega^\infty$, we see that
\begin{equation*}
a(u_1^c - u_2^c + (\chi_1 - 1) \cdot u^\infty, v) = 0.
\end{equation*}
This implies that $(u_1^c + \chi_1 \cdot u^\infty) - (u_2^c + u^\infty) = c$ for some constant $c \in \mathbb{R}$, completing the proof.
\end{proof}

\section{Analytical expressions for EEG FEM}
\label{analytical_expressions}
\setcounter{figure}{0}
\setcounter{table}{0}
\setcounter{lemma}{0}
In the main paper, we noted that in the case of affine trial functions on tetrahedral meshes we compute the FEM right-hand side using analytical expression. We have to assemble three types of integrals, namely
\begin{align}
I_P =& \int_K \langle\sigma^c \nabla u^\infty, \nabla \varphi \rangle \, dV, \label{patch_integral_eeg_ana_chapter}\\
I_S =& \int_F \langle \sigma^\infty \nabla u^\infty , \eta \rangle \varphi \, dS, \label{surface_integral_eeg_ana_chapter}\\
I_T =& \int_K \langle \sigma \nabla \left(\chi \cdot u^\infty \right), \nabla \varphi\rangle \, dV \label{transition_integral_eeg_ana_chapter},
\end{align} 
where $K$ is some tetrahedron and $F$ is some triangle. We call $I_P$ a \textit{patch integral}, $I_S$ a \textit{surface integral} and $I_T$ a \textit{transition integral}. One can find analytical expressions for the patch integral and the surface integral in \cite{beltrachini_analytic_subtraction}. We slightly modified those expressions and used the techniques described in \cite{beltrachini_analytic_subtraction, wilton_analytic_formulas, graglia_analytic_formulas} to derive an analytical expression for the transition integral\footnote{If the reader is mainly interested in implementing the formulas for an analytic computation of the transition integral, all necessary information is given in the following. But since this work is an extension of \cite{beltrachini_analytic_subtraction, wilton_analytic_formulas, graglia_analytic_formulas}, we strongly recommend first studying these papers.}. To set the stage for these formulas we first have to introduce some notation. In the following, we always assume a dipolar source at position $x_0$ with moment $M$, and we assume $x_0 \not \in F$.

Let $F$ be some triangle. Let $p_1, p_2, p_3$ be an arbitrary numbering of the vertices of $F$. Then set
\begin{align*}
s_1 =& \frac{p_3 - p_2}{\|p_3 - p_2\|}, \gamma_1 = \|p_3 - p_2\|,\\
s_2 =& \frac{p_1 - p_3}{\|p_1 - p_3\|}, \gamma_2 = \|p_1 - p_3\|,\\
s_3 =& \frac{p_2 - p_1}{\|p_2 - p_1\|}, \gamma_3 = \|p_2 - p_1\|.
\end{align*}
Now set $u = s_3, w = \frac{s_1 \times s_2}{\|s_1 \times s_2\|}$ and $v = w \times u$, and let 
\[
m_i = s_i \times w
\]
for $i \in \{1, 2, 3\}$. Then set
\begin{align*}
u_3 =& \langle u, p_3 - p_1\rangle,\\
v_3 =& \langle v, p_3 - p_1 \rangle,\\
u_0 =& \langle u, x_0 - p_1 \rangle,\\
v_0 =& \langle v, x_0 - p_1 \rangle,\\
w_0 =& -\langle w, x_0 - p_1 \rangle.
\end{align*}
Then define 
\begin{align*}
t_1 =& \frac{v_0 (u_3 - \gamma_3) + v_3 (\gamma_3 - u_0)}{\gamma_1},\\
t_2 =& \frac{u_0 v_3 - v_0 u_3}{\gamma_2},\\
t_3 =& v_0.
\end{align*}
and
\begin{align*}
\gamma_1^- =& -\frac{(\gamma_3 - u_0) (\gamma_3 - u_3) + v_0 v_3}{\gamma_1}\\
\gamma_1^+ =& \frac{(u_3 - u_0) (u_3 - \gamma_3) + v_3  (v_3 - v_0)}{\gamma_1}\\
\gamma_2^- =& -\frac{u_3 (u_3 - u_0) + v_3 (v_3 - v_0)}{\gamma_2}\\
\gamma_2^+ =& \frac{u_0 u_3 + v_0 v_3}{\gamma_2}\\
\gamma_3^- =& -u_0\\
\gamma_3^+ =& \gamma_3 - u_0.
\end{align*}
We refer to \cite{beltrachini_analytic_subtraction} for a geometric intuition behind these values (note that there is an error in the formula for $\gamma_2^-$ in \cite{beltrachini_analytic_subtraction}, compare e.g. against \cite{graglia_analytic_formulas}). Essentially, we have constructed a map of the triangle and computed the coordinates of its vertices and parametrizations of its edges. We then set for $1 \leq i \leq 3$
\begin{align}
R_i =& \sqrt{t_i^2 + w_0^2} \nonumber\\
R_i^{\pm} =& \sqrt{t_i^2 +w_0^2 + \left(\gamma_i^{\pm}\right)^2} \nonumber\\
f_i =& \begin{cases}
\ln\left(\frac{R_i^+ + \gamma_i^+}{R_i^- + \gamma_i^-}\right) & \text{if $\gamma_i^- \geq 0$}\\
\ln\left(\frac{R_i^- - \gamma_i^-}{R_i^+ - \gamma_i^+}\right) & \text{if $\gamma_i^- < 0$}
\end{cases} \label{f_i_definition}\\
\beta_i =& \arctan\left(\frac{t_i \gamma_i^+}{R_i^2 + |w_0| R_i^+}\right) \nonumber\\
&- \arctan\left(\frac{t_i \gamma_i^-}{R_i^2 + |w_0| R_i^-}\right) \nonumber\\
\beta =& \beta_1 + \beta_2 + \beta_3.\nonumber
\end{align}
Note in particular that our definition of $f_i$ (see equation \eqref{f_i_definition}) differs from the one given in \cite{beltrachini_analytic_subtraction, wilton_analytic_formulas, graglia_analytic_formulas}, where no case distinction is made and the first case is always taken. Carefully looking at the formulas, we see that $R_i^-$ is the distance of the source position $x_0$ to the point $p_{i +1}$ and $R_i^+$ is the distance from the point $p_{i + 2}$ (with $p_4 := p_1, p_5 := p_2$). Furthermore, $\gamma_i^-$ is the oriented distance from the projection of the source position onto the affine line defined by $p_{i +1}$ and $p_{i + 2}$, where the positive orientation is given by $p_{i + 2} - p_{i +1}$, to the point $p_{i + 1}$. Now assume that the source position lies (approximately) on a line defined by some edge of the triangle and happens to lie on the ray defined by the edge and the positive direction. In this case, we have $R_i^+ \approx -\gamma_i^+$ and $R_i^- \approx -\gamma_i^-$ for this edge, and the computation of $f_i$ using the first case becomes unstable or even undefined. Taking a close look at the derivation of the formulas in \cite{beltrachini_analytic_subtraction, wilton_analytic_formulas, graglia_analytic_formulas}, we see that $f_i$ arises from the computation of the integral
\[
\int_{\gamma_i^-}^{\gamma_i^+} \frac{1}{\sqrt{w_0^2 + t_i^2 + t^2}} \, dt.
\]
A straightforward computation reveals that for $a > 0$ we have
\[
\frac{d}{dx}\left(\frac{1}{2}\ln\left(\frac{\sqrt{x^2 + a} + x}{\sqrt{x^2 + a} - x}\right) \right) = \frac{1}{\sqrt{x^2 + a}},
\]
which can be used for the computation of the aforementioned integral. Expanding the argument of the logarithm by $\frac{\sqrt{x^2 + a} + x}{\sqrt{x^2 + a} + x}$ leads to the first expression for $f_i$ given in equation \eqref{f_i_definition}, while expanding with $\frac{\sqrt{x^2 + a} - x}{\sqrt{x^2 + a} - x}$ leads to the second expression. If $w_0 \neq 0$ or $t_i \neq 0$, the computation of $f_i$ is unproblematic. Otherwise the assumption $x_0 \not \in F$ forces either $\gamma_i^+, \gamma_i^- > 0$ or $\gamma_i^+, \gamma_i^- < 0$. In the first case, the first expression for $f_i$ provides a stable way to compute $f_i$, and in the second case, the second expression for $f_i$ provides a stable way to compute $f_i$. The formula for $f_i$ given above thus avoids problems arising from the case where the source position lies on the line defined by some edge of the triangle.

Now set
\begin{equation*}
a = \begin{pmatrix}
-\frac{1}{\gamma_3} \\
\frac{1}{\gamma_3} \\
0
\end{pmatrix},
b = \begin{pmatrix}
\frac{\frac{u_3}{\gamma_3} - 1}{v_3} \\
-\frac{u_3}{\gamma_3 v_3}\\
\frac{1}{v_3}
\end{pmatrix},
\varphi(u_0, v_0) = \begin{pmatrix}
1\\
0\\
0
\end{pmatrix}
 + u_0 a + v_0 b.
\end{equation*}

We can now give the analytical formulas for the transition integrals (see equation \eqref{transition_integral_eeg_ana_chapter}).

\begin{lemma}
We want to compute the integral
\[
\int_K \langle \sigma \nabla \left(\chi \cdot u^\infty \right), \nabla \varphi\rangle \, dV,
\]
where $K$ is some tetrahedron, $\varphi$ is a nodal Lagrange basis function, and $\chi$ is affine on $K$. Let $F_1, F_2, F_3, F_4$ be the faces of the tetrahedron. Let $1 \leq j \leq 4$. Let $p_1(F_j), p_2(F_J), p_3(F_j)$ be some numbering of the vertices of $F_j$. Then set 
\[
c(F_j) = 
\begin{pmatrix}
\chi(p_1(F_j))\\
\chi(p_2(F_j))\\
\chi(p_3(F_j))
\end{pmatrix}.
\]
Furthermore, let $M(F_j) = (M_0, M_u, M_v) \in \mathbb{R}^{3\times 3}$, where
\begin{align*}
M_0 =& \sign(w_0) \beta w - \sum_{i = 1}^3 f_i m_i,\\
M_{u} =& \sum_{i = 1}^3 \langle u, s_i\rangle \left(f_i t_i s_i - (R_i^+ - R_i^-) m_i\right)\\
& - |w_0| \beta u - \left(\sum_{i = 1}^3 \langle u, m_i\rangle f_i\right) w_0 w,\\
M_{v} =& \sum_{i = 1}^3 \langle v, s_i\rangle \left(f_i t_i s_i - (R_i^+ - R_i^-) m_i\right)\\
& - |w_0| \beta v - \left(\sum_{i = 1}^3 \langle u, m_i\rangle f_i\right) w_0 w,
\end{align*}
where it is understood that all values in the above equations depend on the face $F_j$ and are computed as defined above.
Then set 
\[
I(F_j) = M(F_j) \cdot
\begin{pmatrix}
\langle \varphi(u_0, v_0), c(F_j)\rangle\\
\langle a, c(F_j)\rangle\\
\langle b, c(F_j)\rangle
\end{pmatrix}.
\]
Now let $\mu_K : \widehat{K} \rightarrow K$ be the affine map from the reference tetrahedron to $K$. Let $A \in \mathbb{R}^{3 \times 3}$ be its (constant) jacobian. Let $\widehat{\varphi}_1, \widehat{\varphi}_2, \widehat{\varphi}_3, \widehat{\varphi}_4$ be the shape functions on the reference element and let $G = (\nabla \widehat{\varphi}_1, \nabla \widehat{\varphi}_2, \nabla \widehat{\varphi}_3, \nabla \widehat{\varphi}_4) \in \mathbb{R}^{3 \times 4}$ be the matrix containing their (constant) jacobians. Denote by $\eta(F_j)$ the (constant) unit outer normal of the face $F_j$.

Then we have
\begin{multline*}
 \begin{pmatrix}
\int_K \langle \sigma \nabla \left(\chi \cdot u^\infty \right), \nabla \varphi_1\rangle \, dV\\
\int_K \langle \sigma \nabla \left(\chi \cdot u^\infty \right), \nabla \varphi_2\rangle \, dV\\
\int_K \langle \sigma \nabla \left(\chi \cdot u^\infty \right), \nabla \varphi_3\rangle \, dV\\
\int_K \langle \sigma \nabla \left(\chi \cdot u^\infty \right), \nabla \varphi_4\rangle \, dV
\end{pmatrix}
\\=
\frac{\left(\sigma A^{-T} G\right)^T}{4 \pi \sigma^\infty}
\sum_{j = 1}^4 \langle I(F_j), M \rangle \eta(F_j).
\end{multline*}
\end{lemma}

\medskip

\begin{proof}
We have
\begin{multline*}
 \begin{pmatrix}
\int_K \langle \sigma \nabla \left(\chi \cdot u^\infty \right), \nabla \varphi_1\rangle \, dV\\
\int_K \langle \sigma \nabla \left(\chi \cdot u^\infty \right), \nabla \varphi_2\rangle \, dV\\
\int_K \langle \sigma \nabla \left(\chi \cdot u^\infty \right), \nabla \varphi_3\rangle \, dV\\
\int_K \langle \sigma \nabla \left(\chi \cdot u^\infty \right), \nabla \varphi_4\rangle \, dV
\end{pmatrix}
\\= 
 \begin{pmatrix}
\nabla \varphi_1^T\\
\nabla \varphi_2^T\\
\nabla \varphi_3^T\\
\nabla \varphi_4^T
\end{pmatrix}
\cdot \sigma \cdot \int_{K} \nabla\left(\chi \cdot u^\infty\right)\, dV \\
= \frac{\left(\sigma A^{-T} G\right)^T}{4 \pi \sigma^\infty}
\int_{K} \nabla \left(\chi \cdot \frac{\langle M, x - x_0\rangle}{\|x - x_0\|^3} \right) \, dV,
\end{multline*}
where in the last line we have used the chain rule, the symmetry of $\sigma$, and the definition of $u^\infty$. Now the divergence theorem yields
\begin{multline*}
\int_{K} \nabla \left(\chi \cdot \frac{\langle M, x - x_0\rangle}{\|x - x_0\|^3} \right) \, dV\\
= \sum_{j = 1}^4 \int_{F_j} \chi \cdot \frac{\langle M, x - x_0\rangle}{\|x - x_0\|^3} \cdot \eta(F_j) \, dS\\
= \sum_{j = 1}^4 \eta(F_j) \langle M, \int_{F_j} \chi \cdot \frac{x - x_0}{\|x - x_0\|^3} \, dS \rangle.
\end{multline*}
We thus need to compute integrals of the form
\[
\int_{F} \chi \cdot \frac{x - x_0}{\|x - x_0\|^3} \, dS.
\]
for a triangle $F$. Let $p_1, p_2, p_3$ denote the vertices of this triangle and let $\varphi_i$ be the unique affine function that is $1$ on $p_i$ and zero on the remaining vertices. Since by construction $\chi$ is affine on the face $F$, we have
\[
\chi = \sum_{i = 1}^3 \chi(p_i) \cdot \varphi_i = \langle\begin{pmatrix}
\varphi_1\\
\varphi_2\\
\varphi_3
\end{pmatrix}, c(F)\rangle.
\]
Now set $O = (u, v, w)$ and define
\[
\Phi : \mathbb{R}^2 \rightarrow \mathbb{R}^3; (z_1, z_2) \mapsto
p_1 + O \cdot \begin{pmatrix}
z_1\\
z_2\\
0
\end{pmatrix}.
\]
Let $\widehat{F}' := \conv\{(0, 0), (\gamma_3, 0), (u_3, v_3)\}$. The values defined above were chosen in such a way that $\Phi|_{\widehat{F}'}$ defines a map of the triangle $F$. To make the following computation easier we follow \cite{beltrachini_analytic_subtraction} and introduce an additional translation 
\[
\tau_\rho : \mathbb{R}^2 \rightarrow \mathbb{R}^3; z \mapsto z + \rho,
\]
where $\rho = (u_0, v_0)$ are the coordinates of the projection of the source position $x_0$ onto the plane defined by the triangle. We now set $\widehat{F} = \tau_\rho^{-1}(\widehat{F}')$ and set $\Phi_\rho = \Phi \circ \tau_\rho$. Then $\Phi_\rho|_{\widehat{F}}$ defines a map of $F$.
Note that the gramian of this map is $1$, and hence we can compute
\[
\int_{F} \chi \cdot \frac{x - x_0}{\|x - x_0\|^3} \, dS.
= \int_{\widehat{F}} \chi(\Phi_\rho(z)) \frac{\Phi_\rho(z) - x_0}{\|\Phi_\rho(z) - x_0\|^3} \, dz.
\]
One can easily check that
\[
\begin{pmatrix}
\varphi_1\\
\varphi_2\\
\varphi_3
\end{pmatrix}(\Phi(z_1, z_2))
= (e_1, a, b) \cdot \begin{pmatrix}
1\\
z_1\\
z_2
\end{pmatrix},
\]
where $e_1 = (1, 0, 0)$ is the first unit vector, see e.g. \cite{beltrachini_analytic_subtraction}, and hence
\[
\begin{pmatrix}
\varphi_1\\
\varphi_2\\
\varphi_3
\end{pmatrix}(\Phi_\rho(z_1, z_2))
= \varphi(u_0, v_0) + z_1 a + z_2 b,
\]
where $\varphi(u_0, v_0)$ is the value defined just before the statement of the lemma.
Putting this together, we get
\begin{multline*}
\int_{F} \chi \cdot \frac{x - x_0}{\|x - x_0\|^3} \, dS
= \int_{\widehat{F}} \chi(\Phi_\rho(z)) \frac{\Phi_\rho(z) - x_0}{\|\Phi_\rho(z) - x_0\|^3} \, dz \\
= \int_{\widehat{F}} \frac{\Phi_\rho(z) - x_0}{\|\Phi_\rho(z) - x_0\|^3} \, dz \cdot \langle \varphi(u_0, v_0), c(F)\rangle \\
+  \int_{\widehat{F}} z_1 \frac{\Phi_\rho(z) - x_0}{\|\Phi_\rho(z) - x_0\|^3} \, dz \cdot \langle a, c(F)\rangle \\
+  \int_{\widehat{F}} z_2 \frac{\Phi_\rho(z) - x_0}{\|\Phi_\rho(z) - x_0\|^3} \, dz \cdot \langle b, c(F)\rangle.
\end{multline*}
Hence we need to compute the values
\begin{align*}
M_0 &= \int_{\widehat{F}} \frac{\Phi_\rho(z) - x_0}{\|\Phi_\rho(z) - x_0\|^3} \, dz \\
M_u &= \int_{\widehat{F}} z_1 \frac{\Phi_\rho(z) - x_0}{\|\Phi_\rho(z) - x_0\|^3} \, dz\\
M_v &= \int_{\widehat{F}} z_2 \frac{\Phi_\rho(z) - x_0}{\|\Phi_\rho(z) - x_0\|^3} \, dz.
\end{align*}
The value of $M_0$ reported in the statement of the lemma is then a direct consequence of combining (25), (26), and (30) from \cite{graglia_analytic_formulas} (note that we followed the convention of \cite{beltrachini_analytic_subtraction}, meaning our definition of $w_0$ and the definition of \cite{graglia_analytic_formulas} differ by a factor of $-1$). Similarly, the values of $M_u$ and $M_v$ can be computed by combining (25), (27), (31), and (32) from \cite{graglia_analytic_formulas}. This completes the derivation.
\end{proof}

\section{Choosing integration orders}
\label{integration_order_chapter}
\setcounter{figure}{0}
\setcounter{table}{0}
\setcounter{lemma}{0}
In the main paper it was described that if we want to solve the EEG forward problem and are given a tetrahedral mesh, piecewise affine test functions, and isotropic conductivity in the element containing the source position, we can assemble the localized subtraction FEM right-hand side analytically. If these conditions for the EEG forward problem are not met, e.g. if we use a hexahedral mesh, or if we want to solve the MEG forward problem, we resort to numerical integration. This of course raises the question of choosing suitable integration orders, which we want to discuss in this section.

Let $K$ be a mesh element and let $F$ be a face of a mesh element. In the EEG case, we are concerned with the integrals
\begin{align}
I_P =& \int_K \langle\sigma^c \nabla u^\infty, \nabla \varphi \rangle \, dV, \label{patch_integral_eeg}\\
I_S =& \int_F \langle \sigma^\infty \nabla u^\infty , \eta \rangle \varphi \, dS, \label{surface_integral_eeg}\\
I_T =& \int_K \langle \sigma \nabla \left(\chi \cdot u^\infty \right), \nabla \varphi\rangle \, dV \label{transition_integral_eeg},
\end{align}
where $\varphi$ is a test function. We call $I_P$ a \textit{patch integral}, $I_S$ a \textit{surface integral} and $I_T$ a \textit{transition integral}. In the MEG case, we are concerned with the integrals
\begin{align}
F_P =& \int_K \sigma^c(y) \nabla u^\infty(y) \times \frac{x - y}{\|x - y\|^3} \, dV(y), \label{patch_flux_meg}\\
F_S =& \int_F \sigma^\infty u^\infty(y) \cdot \eta(y) \times \frac{x - y}{\|x - y\|^3} \, dS(y), \label{surface_flux_meg}\\
F_T =& \int_K \sigma(y) \nabla\left(\chi \cdot u^\infty\right)(y) \times \frac{x - y}{\|x - y\|^3} \, dV(y), \label{transition_flux_meg},
\end{align}
where we call $F_P$ a \textit{patch flux}, $F_S$ a \textit{surface flux} and $F_T$ a \textit{transition flux}.

In \cite{beltrachini_analytic_subtraction}, the author investigated the error introduced by numerical quadrature for the patch integrals and surface integrals. In particular, the author shows experimentally that the central quantity for assessing the accuracy of the numerical quadrature is the ratio $d/a$, where $d$ is the distance of the source position to the element under consideration and $a$ is the edge length of this element. Here, smaller ratios demand higher integration orders to achieve the same accuracy.

This observation is indeed a mathematical necessity. When performing numerical quadrature, one typically chooses a map of the element $K$ (resp. the face $F$) to pull back the integral to the reference element, and then performs the numerical integration on the reference element. A straightforward computation now shows that when scaling the source position and the element $K$ (resp. the face $F$) with a scalar $\lambda \neq 0$, the integrands \eqref{patch_integral_eeg}-\eqref{transition_integral_eeg}, when pulled back to the reference element, are scaled by a factor of $1/\lambda$, at least for affine and multilinear test functions. This implies that the analytic and numeric integrals are also scaled by a factor of $1/\lambda$. This in turn implies that the relative error introduced by numerical integration is scale-invariant. If the coil position is also scaled, another straightforward computation shows that the integrands \eqref{patch_flux_meg}-\eqref{transition_flux_meg} on the reference element are scaled by a factor of $1/\lambda^2$, again implying scale-invariance of the relative error introduced by numerical integration.

This observation enables us to specify rules for choosing integration orders independently of a concrete mesh since the scale-invariance of the relative error enables us to investigate the accuracy of different integration orders at a reference scale.

Our tests showed that, when compared to the dipole distance or the maximum edge length, the element shape and the dipole orientation only had a minor influence on the relative integration error. Among the elements of fixed maximum edge length, the equilateral ones tended to produce the largest errors, which we suppose is a consequence of them having the lowest density of quadrature points\footnote{The underlying heuristic here is that a higher density of quadrature points should enable a more accurate numerical integration. The density of quadrature points is given by the number of quadrature points divided by the element volume. When performing numerical integration by pulling back to a reference element, the number of quadrature points for a given integration order is independent of the element shape, and hence the quadrature point density for a given integration order only depends on the reciprocal of the element volume.}. We will thus derive quadrature rules by investigating the relative error introduced by numerical integration for a single equilateral element of side length $a = 1$ for different dipoles at distance $d = d/a$.

More concretely, we used equilateral tetrahedral and hexahedral elements and computed the relative error introduced by numerical integration for different integration orders and different values of $\frac{d}{a}$. In particular, we used the tetrahedron with vertices
\begin{align*}
p_1 &= \begin{pmatrix}
0, & -\frac{1}{2}, & -\frac{\sqrt{3}}{6}
\end{pmatrix},\\
p_2 &= \begin{pmatrix}
0, & \frac{1}{2}, & -\frac{\sqrt{3}}{6}
\end{pmatrix},\\
p_3 &= \begin{pmatrix}
0, & 0, & -\frac{\sqrt{3}}{3}
\end{pmatrix},\\
p_4 &= \begin{pmatrix}
-\frac{\sqrt{6}}{3}, & 0, & 0
\end{pmatrix},
\end{align*}
and for the surface integrals the triangle with vertices $p_1, p_2, p_3$. 

Furthermore, we took the cube with vertices
\begin{align*}
q_1 &= \begin{pmatrix}
-1, & -0.5, -0.5
\end{pmatrix}\\
q_2 &= \begin{pmatrix}
0, & -0.5, & -0.5
\end{pmatrix}\\
q_3 &= \begin{pmatrix}
-1, & 0.5, & -0.5
\end{pmatrix}\\
q_4 &= \begin{pmatrix}
0, & 0.5, & -0.5
\end{pmatrix}\\
q_5 &= \begin{pmatrix}
-1, & -0.5, & 0.5
\end{pmatrix}\\
q_6 &= \begin{pmatrix}
0, & -0.5, & 0.5
\end{pmatrix}\\
q_7 &= \begin{pmatrix}
-1, & 0.5, & 0.5
\end{pmatrix}\\
q_8 &= \begin{pmatrix}
0, & 0.5, & 0.5
\end{pmatrix}
\end{align*}
and for the surface integrals the square possessing the vertices $q_2, q_4, q_6, q_8$. We used the source position
\[
x_0 = d \cdot
\begin{pmatrix}
1, & 0.1, & 0.1
\end{pmatrix}.
\]
Note that we did not choose $(1, 0, 0)$, since for this position some integration orders perform unusually well due to what we suppose are symmetry reasons. Furthermore, we used the point
\[
\begin{pmatrix}
20, 0, 0
\end{pmatrix}
\]
as a coil position for the computation of the magnetic field. Our experiments suggest that, as long as the distance of the source to the element is significantly smaller than the distance of the coil to the element, the concrete coil position only has a minor influence on the error. Furthermore, with increasing distance of the coil to the element, the error tended to decrease. Considering the setup of modern MEG machines, it thus seems justified to use a single coil at a distance of $2$cm from the element in our experiments. Relative errors were computed by comparing against pseudo-analytic integrals obtained by using an integration order of 50. For the transition integrals and fluxes, we defined $\chi$ to be the affine resp. the multilinear function which is 1 on all vertices belonging to the face closest to the dipole and 0 on the remaining vertices. The code used for the computations in this section was implemented as a DUNE module\footnote{\url{https://github.com/MalteHoel/local_integration_test}}.

Since we perform two vertex extensions, we can expect the distance-edge length ratios for the transition and surface integrals to be distributed around $2$. To check our intuition, we took \textit{mesh\_init} from the results section of the main paper and $1000$ dipoles at $0.99$ eccentricity. For each of these dipoles, we built the corresponding patch and computed for every element in the transition region the ratio d/a. The result is shown in Figure \ref{distance_edge_length_ratio_distribution}. We computed $a$ as the longest edge length and $d$ as the minimum value of the distances of the source position to the corners and the centers of the faces of the tetrahedron.

\begin{figure}[]
\centerline{\includegraphics[]{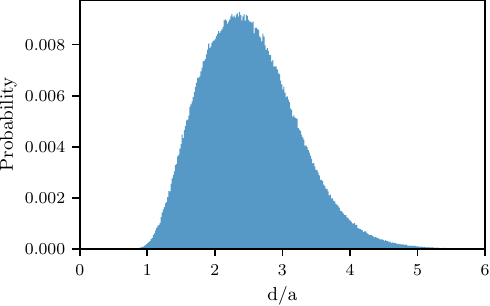}}
\caption{Histogram visualizing the relative frequency of distance-edge length ratios during the computation of surface and transition integrals (\eqref{surface_flux_meg} and \eqref{transition_flux_meg}) in a tetrahedral mesh. To generate this image, $1000$ dipoles at an eccentricity of $0.99$ were considered.} 
\label{distance_edge_length_ratio_distribution}
\end{figure}

We see that the distribution peaks around $2$, and only a very small fraction of element-source pairs have a ratio less than $1$, as we expected. Hence if we choose an integration order that is accurate for a ratio of $1$ (and slightly below $1$), we can expect the relative error for the surface and transition integrals to be  small.

We first want to derive integration orders for the MEG fluxes \eqref{patch_flux_meg}-\eqref{transition_flux_meg} in the case of a tetrahedral mesh. In Figures \ref{tet_meg_surface} and \ref{tet_meg_transition} we see the relative errors for different integration orders for the surface \eqref{surface_flux_meg} and transition fluxes \eqref{transition_flux_meg}. Based on this, we suggest using an integration order of $6$ for the surface integrals and an integration order of $5$ for the transition integrals.

\begin{figure}[]
\centerline{\includegraphics[]{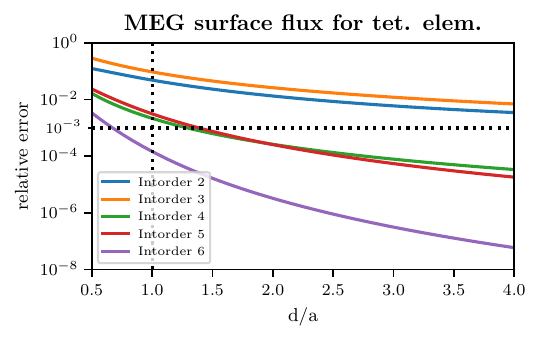}}
\caption{Relative error introduced by numerical integration for the MEG surface flux \eqref{surface_flux_meg} for a tetrahedral element for different integration orders. On the $x$-axis we have the distance-edge length ratio and on the $y$-axis we have the relative error.} 
\label{tet_meg_surface}
\end{figure}

\begin{figure}[]
\centerline{\includegraphics[]{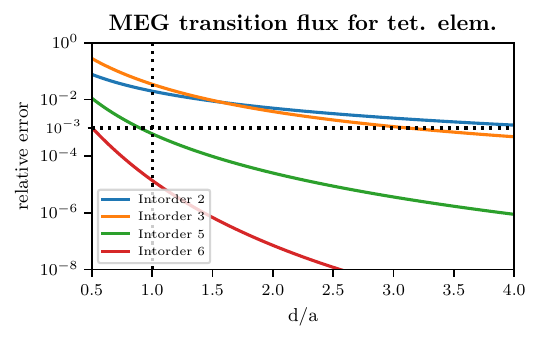}}
\caption{Relative error introduced by numerical integration of the MEG transition flux \eqref{transition_flux_meg} for a tetrahedral element for different integration orders. On the $x$-axis we have the distance-edge length ratio and on the $y$-axis we have the relative error} 
\label{tet_meg_transition}
\end{figure}

For the patch fluxes, we cannot expect d/a to be nicely bounded below. Instead, we suggest choosing an individual integration order for each element-source combination. More concretely, once d/a has been computed, one can use the results shown in Figure \ref{tet_meg_patch} to choose an appropriate integration order. This leads to the rule already described in table \ref{tet_meg_patch_table}. This rule gives accurate integration orders as long as $\frac{d}{a} \geq \frac{1}{6}$, which should suffice for all meshes and sources used in practical applications.

\begin{figure}[]
\centerline{\includegraphics[]{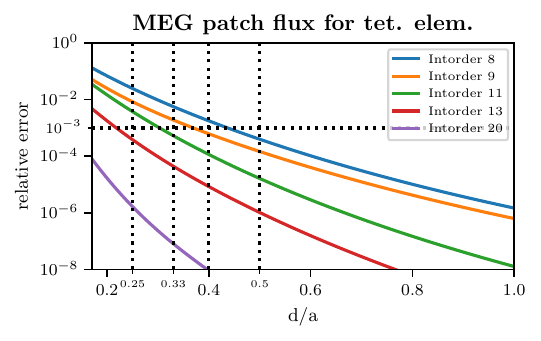}}
\caption{Relative error introduced by numerical integration of the MEG patch flux for \eqref{patch_flux_meg} a tetrahedral element for different integration orders. On the $x$-axis we have the distance-edge length ratio and on the $y$-axis we have the relative error} 
\label{tet_meg_patch}
\end{figure}

\addtolength{\tabcolsep}{-5pt}
\begin{table}
\caption{Tetrahedral MEG patch flux integration order}
\label{tet_meg_patch_table}
\begin{tabular}{|p{56pt}|p{36.25pt}|p{36.25pt}|p{36.25pt}|p{36.25pt}|p{36.25pt}|}
\hline
\small \mbox{$\rule{0cm}{0.3cm}\frac{\text{distance}}{\text{edge length}}$ $\left(\frac{d}{a}\right)$} & \small\mbox{$\frac{d}{a} \geq 0.5$} & \small \mbox{$\frac{d}{a} \geq 0.4$} & \small \mbox{$\frac{d}{a} \geq 0.33$} &  \small \mbox{$\frac{d}{a} \geq 0.25$} & \small \mbox{$\frac{d}{a} \geq 0.17$} \\ \hline
\rule{0cm}{0.3cm}\small Intorder		& \small 8	& \small 9 &  \small 11 & \small 13 & \small 20\\ \hline
\end{tabular}
\end{table}
\addtolength{\tabcolsep}{5pt}

We now consider the hexahedral case. In this case, we need to specify integration orders for the EEG and the MEG case. We again perform two vertex extensions to construct the patch. Note that we can expect a larger lower bound for the $d/a$-ratio for the surface and transition integrals since hexahedral meshes are typically generated from voxel-based segmentation followed by a potential geometry adaption. To check this assumption, we constructed a geometry-adapted 1 mm hexahedral mesh with a node shift factor of $0.33$, as described in \cite{wolters_geometry_adapted_hexahedra}. Similar to the tetrahedral case, we then took 1000 dipoles at an eccentricity of $0.99$ and constructed the corresponding patches and transition regions. For each dipole and each element inside the corresponding transition region, we then computed the ratio $d/a$, using the same approach as in the tetrahedral case. The result is shown in Figure \ref{hex_distance_edge_length_ratio_distribution}. We see that the distribution peaks at $2.5$, and that the vast majority of elements have a ratio above $1.5$. Moreover, not a single element had a ratio  below $1$. If we thus choose integration orders which are accurate up to a ratio of $1$, we can expect the numerical integration to only produce a small relative error for the surface and transition integrals (\ref{surface_integral_eeg}), (\ref{transition_integral_eeg}) and fluxes (\ref{surface_flux_meg}), (\ref{transition_flux_meg}).

Furthermore, it was discussed in the main paper that physiologically relevant sources have a distance of at least $1$ mm to the conductivity jump. Under the assumption that hexahedral meshes used in practice have an edge length of at most $2$ mm, an integration order that is accurate up to a ratio of $0.5$ thus suffices for an accurate computation of the patch integral \eqref{patch_integral_eeg} and patch flux \eqref{patch_flux_meg}.\footnote{Note however that we strongly recommend using a 1 mm mesh over a 2 mm mesh.}

\begin{figure}
\centerline{\includegraphics[]{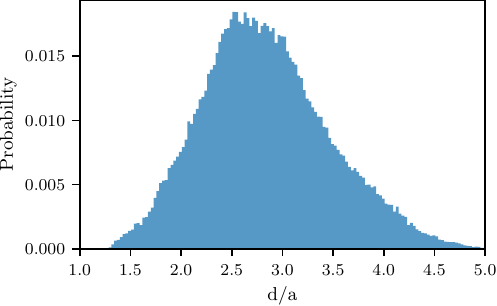}}
\caption{Histogram visualizing the relative frequency of distance-edge length ratios during the computation of surface and transition integrals (\eqref{surface_flux_meg} and \eqref{transition_flux_meg}) in a geometry-adapted hexahedral mesh. To generate this image, $1000$ dipoles at an eccentricity of $0.99$ were considered.}
\label{hex_distance_edge_length_ratio_distribution}
\end{figure}

Similar to the tetrahedral case, we computed the relative error introduced by numerical integration on a single hexahedral element for different $d/a$-values. Figures \ref{hex_eeg_surface}, \ref{hex_eeg_transition}, and \ref{hex_eeg_patch} show the corresponding results for the EEG case, while Figures \ref{hex_meg_surface}, \ref{hex_meg_transition}, and \ref{hex_meg_patch} show the results for the MEG case. 

Based on the discussion in the previous paragraphs and the results in these figures, we thus suggest choosing an integration order according to table \ref{hexahedral_integration_orders}. Note that we can be conservative when choosing orders since when using the localized subtraction approach the computation times are small enough to not be a serious concern anymore.

\begin{figure}[]
\centerline{\includegraphics[]{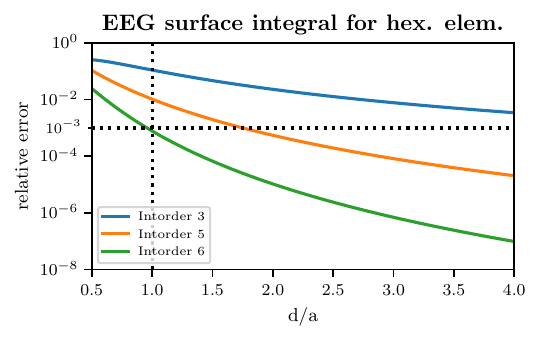}}
\caption{Relative error introduced by numerical integration of the EEG surface integral \eqref{surface_integral_eeg} for a hexahedral element for different integration orders. On the $x$-axis we have the distance-edge length ratio and on the $y$-axis we have the relative error} 
\label{hex_eeg_surface}
\end{figure}

\begin{figure}[]
\centerline{\includegraphics[]{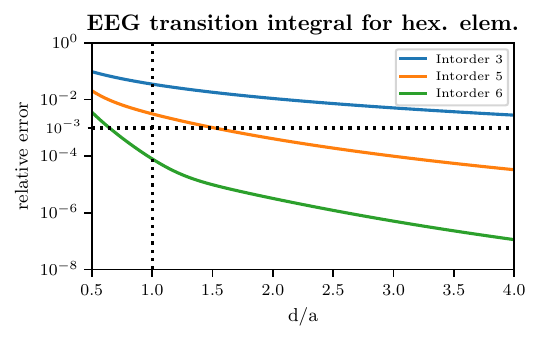}}
\caption{Relative error introduced by numerical integration of the EEG transition integral \eqref{transition_integral_eeg} for a hexahedral element for different integration orders. On the $x$-axis we have the distance-edge length ratio and on the $y$-axis we have the relative error} 
\label{hex_eeg_transition}
\end{figure}

\begin{figure}[]
\centerline{\includegraphics[]{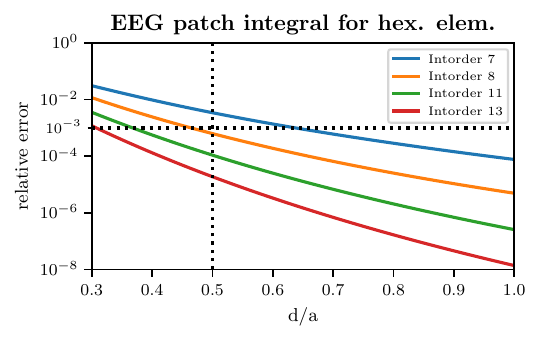}}
\caption{Relative error introduced by numerical integration of the EEG patch integral \eqref{patch_integral_eeg} for a hexahedral element for different integration orders. On the $x$-axis we have the distance-edge length ratio and on the $y$-axis we have the relative error} 
\label{hex_eeg_patch}
\end{figure}

\begin{figure}[]
\centerline{\includegraphics[]{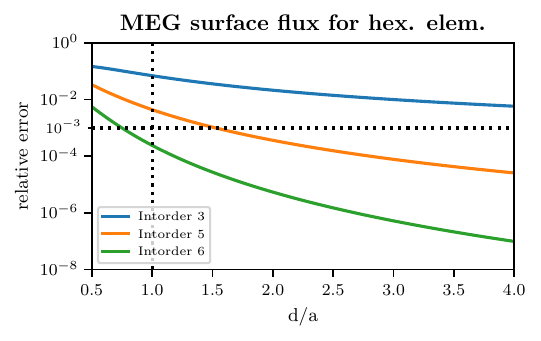}}
\caption{Relative error introduced by numerical integration of the MEG surface flux \eqref{surface_flux_meg} for a hexahedral element for different integration orders. On the $x$-axis we have the distance-edge length ratio and on the $y$-axis we have the relative error} 
\label{hex_meg_surface}
\end{figure}

\begin{figure}[]
\centerline{\includegraphics[]{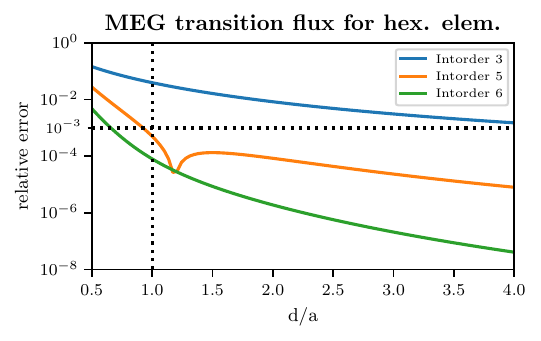}}
\caption{Relative error introduced by numerical integration of the MEG transition flux  \eqref{transition_flux_meg} for a hexahedral element for different integration orders. On the $x$-axis we have the distance-edge length ratio and on the $y$-axis we have the relative error} 
\label{hex_meg_transition}
\end{figure}

\begin{figure}[]
\centerline{\includegraphics[]{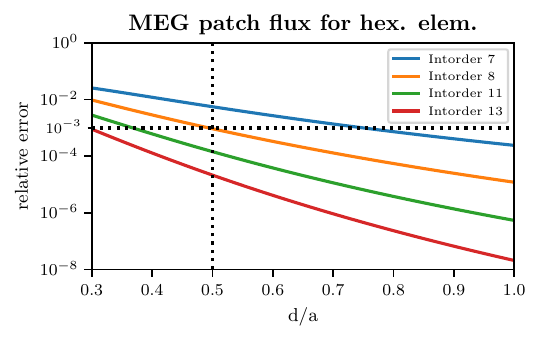}}
\caption{Relative error introduced by numerical integration of the MEG patch flux  \eqref{patch_flux_meg} for a hexahedral element for different integration orders. On the $x$-axis we have the distance-edge length ratio and on the $y$-axis we have the relative error} 
\label{hex_meg_patch}
\end{figure}

\begin{table}
\caption{Hexahedral integration orders}
\label{hexahedral_integration_orders}
\begin{tabular}{|p{50.65pt}|p{50.65pt}|p{50.65pt}|p{50.65pt}|}
\hline
& Surface & Transition & Patch \\ \hline
EEG & 6 & 6 & 8 \\ \hline
MEG & 6 & 6 & 8 \\ \hline
\end{tabular}
\end{table}

\section{Vertex Extension Comparison}
\label{vertex_extension_comparison_appendix_section}
\setcounter{figure}{0}
\setcounter{table}{0}
\setcounter{lemma}{0}
In the main paper, we investigated the influence of the number of vertex extensions on the accuracy of the localized subtraction approach. In particular, we referred to this appendix for the figures showing the corresponding error curves for \textit{mesh\_brain} and \textit{mesh\_skin}. The results for \textit{mesh\_brain} can be seen in Figure \ref{vertex_extension_comparison_brain_refined_tangential_eeg}, and the results for \textit{mesh\_skin} can be seen in Figure \ref{vertex_extension_comparison_skin_refined_tangential_eeg}.

\begin{figure}[]
\centerline{\includegraphics[]{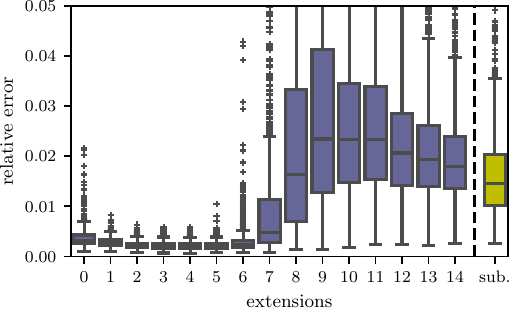}}
\caption{Relative error in the EEG case for 1000 tangential dipoles at 0.99 eccentricity for different numbers of vertex extensions during patch construction, computed using \textit{mesh\_brain} (see Figure \ref{meshes_refinement}\,(b)). The rightmost yellow boxplot shows the errors when employing the analytical subtraction approach from \cite{beltrachini_analytic_subtraction}.}
\label{vertex_extension_comparison_brain_refined_tangential_eeg}
\end{figure}

\begin{figure}[]
\centerline{\includegraphics[]{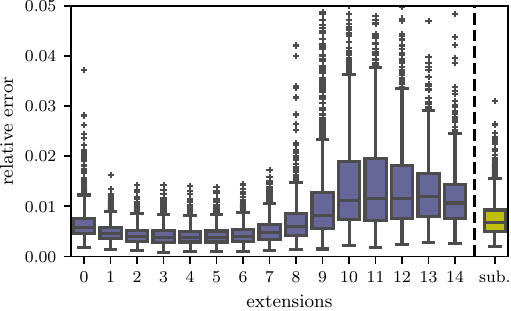}}
\caption{Relative error in the EEG case for 1000 tangential dipoles at 0.99 eccentricity for different numbers of vertex extensions during patch construction, computed using \textit{mesh\_skin} (see Figure \ref{meshes_refinement}\,(c)). The rightmost yellow boxplot shows the errors when employing the analytical subtraction approach from \cite{beltrachini_analytic_subtraction}.}
\label{vertex_extension_comparison_skin_refined_tangential_eeg}
\end{figure}

We also performed the vertex extension comparison for the MEG forward problem. The corresponding results for \textit{mesh\_init} can be seen in Figure \ref{vertex_extension_comparison_init_tangential_meg}. We see an increase in accuracy when going from 0 to 2 extensions. After that, the errors barely change. Similar to the EEG case, the initial error decrease is due to smaller patches forcing more singular behavior onto the correction potential. But in contrast to the EEG case, the MEG errors do not increase again as the patch becomes larger. This again demonstrates that the secondary magnetic field due to a dipolar source is predominantly influenced by the volume currents in the close proximity of the source.  Performing the corresponding experiment for \textit{mesh\_brain} and \textit{mesh\_skin} leads to the same qualitative results.

\begin{figure}[]
\centerline{\includegraphics[]{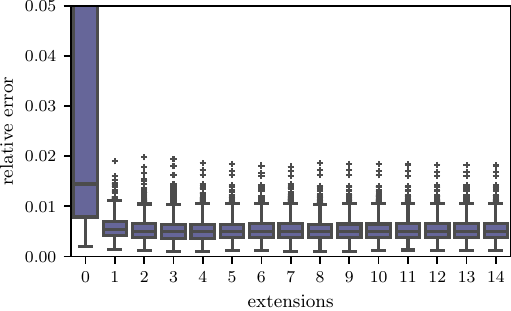}}
\caption{Relative error in the MEG case for 1000 tangential dipoles at 0.99 eccentricity for different numbers of vertex extensions during patch construction, computed using \textit{mesh\_init} (see Figure \ref{meshes_refinement}\,(a)).}
\label{vertex_extension_comparison_init_tangential_meg}
\end{figure}

\section{Accuracy comparison for different potential approaches}
\label{accuracy_comparison_appendix_section}
\setcounter{figure}{0}
\setcounter{table}{0}
\setcounter{lemma}{0}
In the main paper, we showed the accuracy comparison of the localized subtraction approach against the multipolar Venant and analytical subtraction approaches in the EEG case for radial sources and in the MEG cases for tangential sources in \textit{mesh\_init}. To keep the presentation focused, we moved the figures showing the remaining comparisons into this appendix. The comparison for tangential sources in the EEG case in \textit{mesh\_init} can be found in Figure \ref{accuracy_comparison_eeg_tangential_mesh_init}. Figures \ref{accuracy_comparison_eeg_radial_mesh_brain}, \ref{accuracy_comparison_eeg_tangential_mesh_brain} and \ref{accuracy_comparison_meg_tangential_mesh_brain} show the results for \textit{mesh\_brain}. Figures  \ref{accuracy_comparison_eeg_radial_mesh_skin}, \ref{accuracy_comparison_eeg_tangential_mesh_skin} and \ref{accuracy_comparison_meg_tangential_mesh_skin} show the results for \textit{mesh\_skin}.

\begin{figure}[]
\centerline{\includegraphics[]{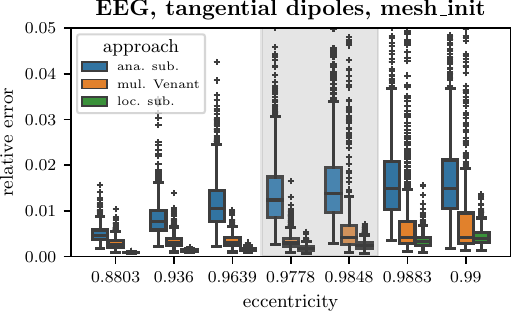}}
\caption{Accuracy comparison of EEG forward simulations using the analytical subtraction, multipolar Venant, and localized subtraction potential approaches for tangential dipoles at different eccentricities using \textit{mesh\_init} (see Figure \ref{meshes_refinement}\,(a)). The $y$-axis shows the relative error. The physiologically relevant sources at 1-2 mm distance from the CSF are highlighted.} 
\label{accuracy_comparison_eeg_tangential_mesh_init}
\end{figure}

\begin{figure}[]
\centerline{\includegraphics[]{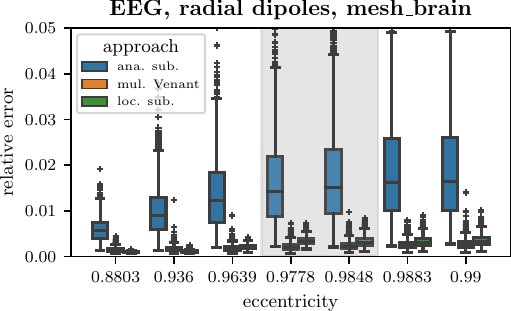}}
\caption{Accuracy comparison of EEG forward simulations using the analytical subtraction, multipolar Venant, and localized subtraction potential approaches for radial dipoles at different eccentricities using \textit{mesh\_brain} (see Figure \ref{meshes_refinement}\,(b)). The $y$-axis shows the relative error. The physiologically relevant sources at 1-2 mm distance from the CSF are highlighted.} 
\label{accuracy_comparison_eeg_radial_mesh_brain}
\end{figure}

\begin{figure}[]
\centerline{\includegraphics[]{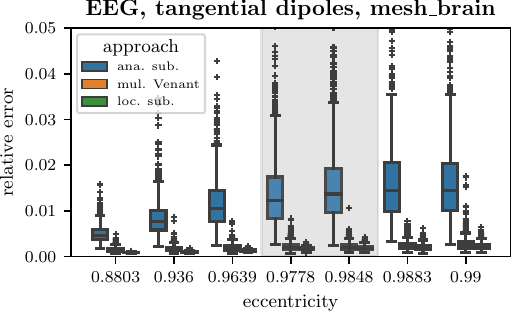}}
\caption{Accuracy comparison of EEG forward simulations using the analytical subtraction, multipolar Venant, and localized subtraction potential approaches for tangential dipoles at different eccentricities using \textit{mesh\_brain} (see Figure \ref{meshes_refinement}\,(b)). The $y$-axis shows the relative error. The physiologically relevant sources at 1-2 mm distance from the CSF are highlighted.} 
\label{accuracy_comparison_eeg_tangential_mesh_brain}
\end{figure}

\begin{figure}[]
\centerline{\includegraphics[]{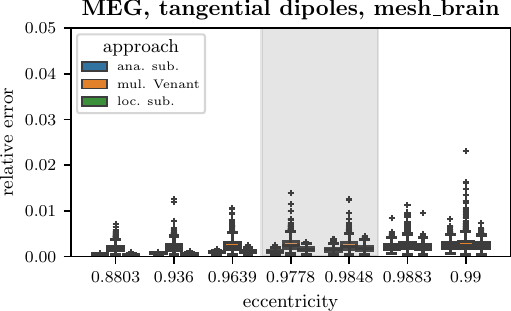}}
\caption{Accuracy comparison of MEG forward simulations using the analytical subtraction, multipolar Venant, and localized subtraction potential approaches for tangential dipoles at different eccentricities using \textit{mesh\_brain} (see Figure \ref{meshes_refinement}\,(b)). The $y$-axis shows the relative error. The physiologically relevant sources at 1-2 mm distance from the CSF are highlighted.} 
\label{accuracy_comparison_meg_tangential_mesh_brain}
\end{figure}

\begin{figure}[]
\centerline{\includegraphics[]{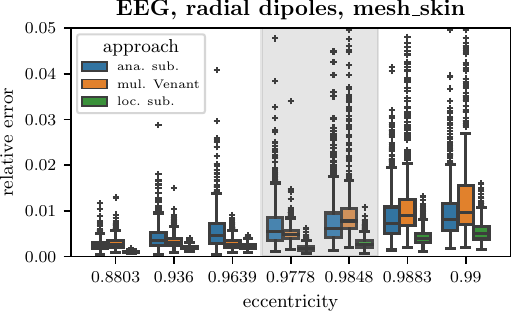}}
\caption{Accuracy comparison of EEG forward simulations using the analytical subtraction, multipolar Venant, and localized subtraction potential approaches for radial dipoles at different eccentricities using \textit{mesh\_skin} (see Figure \ref{meshes_refinement}\,(c)). The $y$-axis shows the relative error. The physiologically relevant sources at 1-2 mm distance from the CSF are highlighted.} 
\label{accuracy_comparison_eeg_radial_mesh_skin}
\end{figure}

\begin{figure}[]
\centerline{\includegraphics[]{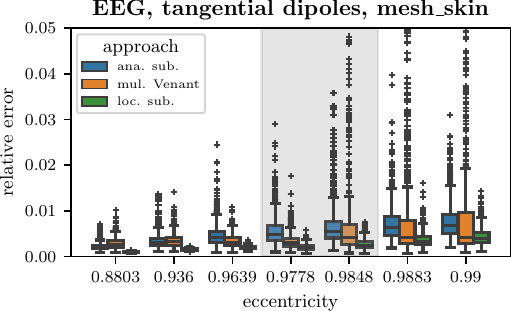}}
\caption{Accuracy comparison of EEG forward simulations using the analytical subtraction, multipolar Venant, and localized subtraction potential approaches for tangential dipoles at different eccentricities using \textit{mesh\_skin} (see Figure \ref{meshes_refinement}\,(c)). The $y$-axis shows the relative error. The physiologically relevant sources at 1-2 mm distance from the CSF are highlighted.} 
\label{accuracy_comparison_eeg_tangential_mesh_skin}
\end{figure}

\begin{figure}[]
\centerline{\includegraphics[]{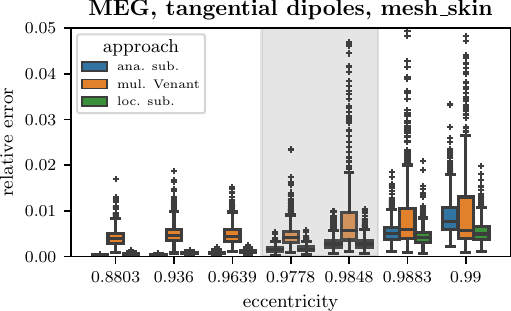}}
\caption{Accuracy comparison of MEG forward simulations using the analytical subtraction, multipolar Venant, and localized subtraction potential approaches for tangential dipoles at different eccentricities using \textit{mesh\_skin} (see Figure \ref{meshes_refinement}\,(c)). The $y$-axis shows the relative error. The physiologically relevant sources at 1-2 mm distance from the CSF are highlighted.} 
\label{accuracy_comparison_meg_tangential_mesh_skin}
\end{figure}

%%%%%%%%%%%%%%%%%%%%%%%%%%%%%%%%%%%%%%%%%%
%%%%%%%%%%%%%%%%%%%%%%%%%%%%%%%%%%%%%%%%%%
%%%%%%%%%%%%%%%%%%%%%%%%%%%%%%%%%%%%%%%%%%
% Bibliography
%%%%%%%%%%%%%%%%%%%%%%%%%%%%%%%%%%%%%%%%%%
%%%%%%%%%%%%%%%%%%%%%%%%%%%%%%%%%%%%%%%%%%
%%%%%%%%%%%%%%%%%%%%%%%%%%%%%%%%%%%%%%%%%%

\bibliography{sources}
\bibliographystyle{ieeetr}

\end{document}